\documentclass[journal]{IEEEtran}

\usepackage{cite}
\usepackage{algorithmic}
\usepackage{graphicx}
\usepackage{textcomp}

\usepackage{microtype}
\usepackage{epsfig}
\usepackage{epstopdf}
\usepackage[cmex10]{amsmath}
\usepackage{mathtools}
\mathtoolsset{showonlyrefs=true}
\usepackage{amssymb}
\usepackage{amsfonts}
\usepackage[normalem]{ulem}
\usepackage{subcaption}
\usepackage[ruled]{algorithm2e}

\usepackage{amsthm}

\usepackage[usenames,dvipsnames]{xcolor}

\DeclareMathOperator{\Lie}{\mathcal{L}}
\DeclareMathOperator{\Stemp}{\mathcal{S}}

\DeclareMathOperator{\Atemp}{\mathcal{A}}
\newcommand{\SL}[3]{\text{$\sideset{^{#2}}{_{#1}^{#3}}\Stemp$}}

\newcommand{\A}[4]{\text{$\sideset{^{#2}_{#3}}{_{#1}^{#4}}\Atemp$}}
\DeclareMathOperator{\dLtemp}{\mathcal{G}}
\newcommand{\dL}[3]{\text{$\sideset{^{#2}}{_{#1}^{#3}}\dLtemp$}}
\newcommand{\dnx}[1]{\text{$\partial^{#1}_x$}}
\DeclareMathOperator{\W}{\mathcal{W}}

\theoremstyle{definition}
\newtheorem{theorem}{Theorem}
\newtheorem{proposition}{Proposition}
\newtheorem{corollary}{Corollary}
\newtheorem{lemma}{Lemma}

\newtheorem{assumption}{Assumption}
\newtheorem{definition}{Definition}
\newtheorem{problem}{Problem}

\theoremstyle{remark}
\newtheorem{remark}{Remark}

\begin{document}

This article has been accepted for publication by IEEE Transactions on Automatic Control.

\vspace{1cm}

The manuscript included in this file is the open access accepted version. 

\vspace{1cm}

This open access version is released on arXiv in accordance with the IEEE copyright agreement.

\vspace{1cm}

The final version will be available at IEEE Xplore.

\title{Stochastic Relative Degree and Path-wise Control of Nonlinear Stochastic Systems}
\author{Alberto Mellone, \IEEEmembership{Student Member, IEEE}, and Giordano Scarciotti, \IEEEmembership{Senior Member, IEEE}
\thanks{*This work is partially supported by Shanghai University under the ``Foreign Expert Program'' of Shanghai University (No. 21WZ0109).}%
\thanks{Alberto Mellone and Giordano Scarciotti are with the Department of Electrical and Electronic \mbox{Engineering}, Imperial College London, London, SW7 2AZ, UK. {\tt\small \{a.mellone18,g.scarciotti\}@imperial.ac.uk}}}

\maketitle

\begin{abstract}
We address the path-wise control of systems described by a set of nonlinear stochastic differential equations. For this class of systems, we introduce a notion of stochastic relative degree and a change of coordinates which transforms the dynamics to a stochastic normal form. The normal form is instrumental for the design of a state-feedback control which linearises and makes the dynamics deterministic. We observe that this control is \emph{idealistic}, \emph{i.e.} it is not practically implementable because it employs a feedback of the Brownian motion (which is never available) to cancel the noise. Using the idealistic control as a starting point, we introduce a hybrid control architecture which achieves \emph{practical} path-wise control. This hybrid controller uses measurements of the state to perform periodic compensations for the noise contribution to the dynamics. We prove that the hybrid controller retrieves the idealistic performances in the limit as the compensating period approaches zero. We address the problem of asymptotic output tracking, solving it in the idealistic and in the practical framework. We finally validate the theory by means of a numerical example. 
\end{abstract}


\section{Introduction}
A point of departure in the study of nonlinear deterministic systems is the definition of the relative degree of the system and, consequently, of a change of coordinates that is able to transform the differential equations in a so-called normal form that makes analysis and control easier. These ideas were first introduced in the seminal work \cite{Isidori1981}, where the authors solved the problem of static state-feedback non-interacting control. The theory of normal forms was later addressed in \cite{Zeitz1983} and \cite{Bestle1983} for the control and observation of time-varying nonlinear systems, and a systematic overview of normal forms was given in \cite{Krener1987}. The problem of feedback linearisation of single-input single-output and multi-input systems was introduced in \cite{Brockett1978} and in \cite{Jakubczyk1980}, respectively, and a systematic procedure to find the feedback-linearising control was provided in \cite{Su1982} and \cite{Hunt1983}. In \cite{Byrnes1984} the notion of zero dynamics was introduced and later employed in \cite{Byrnes1988} to tackle the problem of asymptotic stabilisation of nonlinear systems.

In this paper we introduce the notion of stochastic relative degree to develop a theory of path-wise feedback control for a general class of systems described by nonlinear stochastic differential equations. The advantage of using this stochastic framework lies in the fact that it allows to address uncertainties characterised by probabilistic properties which cannot be captured by the classical deterministic robust designs. Modelling in the stochastic framework is flexible and lends itself to mechanical systems (\emph{e.g.} the quarter-car model), electro-mechanical systems (\emph{e.g.} the suspended gyro), aerospace systems (\emph{e.g.} the satellite dynamics) and mathematical finance. A survey of these applications can be found in \cite[Section 1.9]{Damm2004} and references therein. Stochastic differential equations are also at the basis of methodological applications such as stochastic $H_\infty$ control (see \cite{Hinrichsen1998}, \cite{Gershon2013} and \cite{Hua2018}), filtering and optimal control (see \cite{Yong1999} and \cite{Oksendal2003}).

Some notions of normal forms for stochastic systems have been introduced in the literature. For example, in \cite{Arnold1998} and \cite{Arnold2003}, Stratonovich calculus was used to obtain a normal form for purely diffusive processes, while \cite{Gaeta1999} employed coordinate changes to introduce symmetries for stochastic differential equations. One of the first works suggesting the convenience of a normal form for control of stochastic systems is \cite{Ladhiri1996}, where the change of coordinates makes the dynamics quasi-linear.  In \cite{Pan2001} and \cite{Pan2002}, similar coordinate transformations are proposed in order to reduce the system dynamics to canonical forms that are amenable to specific control strategies. Specifically, strict-feedback noise-prone dynamics are obtained, which allow achieving optimal globally stabilising back-stepping controllers as proposed in \cite{Pan1999}.

The differences between the normal forms mentioned above and the one we introduce in this paper are both technical and in scope. From a technical viewpoint, we propose a coordinate projection and a feedback control that \emph{annihilate} the noise in the linearised coordinates. From an objective viewpoint, the stochastic relative degree and the normal form are mere instrumental means for the design of a practical hybrid controller which achieves output tracking in a path-wise fashion. In other words, the controller that we develop compensates for each \emph{specific} realisation of the stochastic disturbance.

More generally, by defining a new stochastic normal form, the goal of this paper is to address the path-wise control of stochastic systems described by a general class of nonlinear stochastic differential equations. We show that the implementation of feedback laws that perfectly linearise the system dynamics in a new set of coordinates comes with insurmountable causality issues, hence the attribute \emph{idealistic} with which we refer to these control laws. In fact, they employ a feedback of the noise, which is not practically available. However, with these idealistic controls at hand, we introduce hybrid nonlinear controllers that incorporate a causal estimator of the Brownian motion. This controller is practically realisable and, although achieving only approximate feedback linearisation and tracking, its performance can be arbitrarily improved by tuning an underlying parameter, retrieving the idealistic case as a limit behaviour.

While preliminary work has been published in \cite{Mellone2019ter}, \cite{Mellone2020bis}, \cite{Mellone2020ter}, in this paper we present several additional contributions. The main novelties introduced are as follows. 1) All the results are now proved, which provides a substantial theoretical contribution. 2) By leveraging the It\^o-Stratonovich equivalence, we are able to obtain sharper results in the theory of the stochastic normal form. 3) We provide a characterisation of the solvability of the feedback linearisation problem. 4) The theory of the practical control has been revised and made sharper. 5) An analysis of the control challenges arising when the input appears in the diffusion term of the equation has been added. 6) Practical asymptotic output tracking has been addressed and solved for the first time. 7) A number of technical results regarding uniform asymptotic stability of nonlinear time-varying stochastic systems have been added and proved. For the sake of readability, they have been gathered in the Appendix. 8) A comprehensive example illustrates the path-wise output tracking of a nonlinear stochastic system both in the idealistic and practical scenarios. 

The paper is organised as follows. In Section~\ref{sec:preliminaries} we recall some preliminary notions on stochastic systems. Section~\ref{section:srd_normal_form} introduces the stochastic relative degree and normal form. In Section~\ref{section:linearisation} we address the problem of feedback linearisation in the idealistic framework. In Section~\ref{sec:practical_linearisation} we propose a hybrid controller, which practically approximates the idealistic linearising control. Section~\ref{sec:stab_track} addresses the problem of asymptotic output tracking. In Section~\ref{sec:example} we validate the theory by means of a numerical example. Section~\ref{sec:conclusions} contains some concluding remarks. Finally, technical lemmas which are instrumental to prove some of the results of the paper have been collected and proved in the Appendix.

\textbf{Notation.} The symbol $\mathbb{Z}$ denotes the set of integer numbers, while $\mathbb{R}$ and $\mathbb{C}$ denote the fields of real and complex numbers, respectively; by adding the subscript ``$<0$'' (``$\ge 0$'', ``$0$'') to any symbol indicating a set of numbers, we denote that subset of numbers with negative (non-negative, zero) real part. Where convenient, the symbol $\dnx{n}$ is used as a shorthand for the operator $\partial^n/\partial x^n$, while $\alpha^{(n)}$ indicates the $n$-th time derivative of $\alpha$. The Lie derivative of the smooth scalar function $h(x)$ along the vector field $f(x)$ is denoted by $\Lie_f h(x)$. We use the recursive notation $\Lie_f^kh(x) = \Lie_f\Lie_f^{k-1}h(x)$, with $\Lie^0_fh(x) = h(x)$. Given two smooth vector fields $f(x)$ and $g(x)$, we define the operator $\text{ad}_fg(x) = (\partial_xg(x))f(x) - (\partial_x f(x))g(x)$, and, recursively, $\text{ad}_f^k g(x) = \text{ad}_f\text{ad}_f^{k-1}g(x)$ with $\text{ad}_f^{0}g(x) = g(x)$.  $(\nabla, \mathfrak{A}, \mathfrak{P})$ is a probability space given by the set $\nabla$, the $\sigma$-algebra $\mathfrak{A}$ defined on $\nabla$ and the probability measure $\mathfrak{P}$ on the measurable space $(\nabla, \mathfrak{A})$. A \textit{stochastic process} with state space $\mathbb{R}^n$ is a family $\{x_t,\,t\in\mathbb{R}\}$ of $\mathbb{R}^n$-valued random variables, \textit{i.e.} for every fixed $t\in\mathbb{R}$, $x_t(\cdot)$ is an $\mathbb{R}^n$-valued random variable and, for every fixed $w \in \nabla$, $x_{\cdot}(w)$ is an $\mathbb{R}^n$-valued function of time \cite[Section 1.8]{Arnold1974}. For ease of notation, we often indicate a stochastic process $\{x_t,\,t\in\mathbb{R}\}$ simply with $x_t$ (this is common in the literature, see \textit{e.g.} \cite{Arnold1974}). With a slight abuse of notation, any subscript different from the symbol ``$t$'' indicates the corresponding component of the vector $x_t$, \emph{e.g.} $x_i$ is the $i$-th component of the vector $x_t$. All mappings appearing as integrands in stochastic integrals are assumed to be integrable in the corresponding sense, namely It\^o's (see, \emph{e.g.}, \cite[Definition 3.1.4]{Oksendal2003}) or Stratonovich's (equivalent to deterministic integrability).

\section{Preliminaries}
\label{sec:preliminaries}
In this section we shortly recall the theory of generalised stochastic processes and define differential operators that will be used in the remainder of the paper.

Let $C^\infty_0(\mathbb{R})$ be the space of all infinitely differentiable functions on $\mathbb{R}$ with compact support \cite[Definition 1.2.1]{Hormander1983}. The following definitions characterise the notions of distribution (also known as generalised function), distributional derivative and generalised stochastic process.

\begin{definition}
	\cite[Definition 3.1]{Duistermaat2010} Let $X$ be an open subset of $\mathbb{R}$. A \emph{distribution on} $X$ is a linear form $\psi$ on $C^\infty_0(\mathbb{R})$ that is also continuous in the sense that
	\begin{equation}
		\lim_{j\rightarrow \infty} \psi(\varphi_j) = \psi(\varphi) \qquad \text{as} \qquad \lim_{j\rightarrow \infty}\varphi_j = \varphi \quad \text{in} \quad C^\infty_0(\mathbb{R}).  
	\end{equation}
\end{definition}

\begin{definition}
\cite[Definition 3.1.1]{Hormander1983} For any distribution $\psi$, its \emph{distributional derivative} $\dot{\psi}$ is defined as the distribution that satisfies
\begin{equation}
\dot{\psi}(\varphi) = - \psi (\dot{\varphi}),
	\quad \forall \varphi \in C_0^\infty(\mathbb{R}).
\end{equation}
\end{definition}
  Note that generalised functions have derivatives of all order, which are generalised functions as well. 
\begin{definition}
	\cite[Section 3.2]{Arnold1974} A \emph{generalised stochastic process} is a random generalised function in the sense that a random variable $\psi(\varphi)$ is assigned to every $\varphi\in C^\infty_0$, where $\psi$ is, with probability 1, a generalised function.
\end{definition}
We now look at the Brownian motion as a generalised stochastic process. Therefore, its distributional derivative is always defined \cite[Section 3.2]{Arnold1974}. In particular, the generalised stochastic process given by such a derivative has zero mean value and covariance function given by the generalised function $\delta(t-s),$ $t,s\in\mathbb{R}$, \emph{i.e.} the Dirac delta. Consequently, the derivative of the generalised Brownian motion is the generalised \emph{white noise} \cite[Section 3.2]{Arnold1974}. In the remainder, with a slight abuse of notation, we refer to generalised Brownian motion and generalised white noise omitting the attribute ``generalised" and we denote them by simply $\mathcal{W}_t$ and $\xi_t$, respectively, with $\xi_t = \dot{\mathcal{W}}_t$. It should be emphasised that the just mentioned time derivative is meant in the sense of distributions, and not as the limit of the difference quotient as the increment tends to zero, which instead applies to differentiable functions in the classical sense.

Consider the nonlinear single-input, single-output stochastic system expressed in the shorthand integral notation by
\begin{equation}
\label{eq:sytem_differentials}
\begin{aligned}
dx_t&= (f(x_t) + g(x_t)u)dt + (l(x_t) + m(x_t)u)d\W_t,\\
y_t &= h(x_t),
\end{aligned}
\end{equation}
with $x_t \in \mathbb{R}^n$, $u\in \mathbb{R}$, $y_t \in \mathbb{R}$ and $f:\mathbb{R}^n\rightarrow \mathbb{R}^n$, $g:\mathbb{R}^n\rightarrow \mathbb{R}^n$, $l:\mathbb{R}^n\rightarrow \mathbb{R}^n$, $m:\mathbb{R}^n\rightarrow \mathbb{R}^n$, $h:\mathbb{R}^n\rightarrow \mathbb{R}$ \emph{smooth} functions, \emph{i.e.} they admit continuous partial derivatives of any order. We assume that, for a fixed initial condition $x_{t=0}$, the solution of~\eqref{eq:sytem_differentials} is unique. Note that, in the light of the previous discussion, system~\eqref{eq:sytem_differentials} can be rewritten in the following differential notation
\begin{equation}
	\label{eq:system_derivatives}
	\dot{x}_t = f(x_t) + g(x_t)u + (l(x_t) + m(x_t)u)\xi_t, \quad y_t = h(x_t).
\end{equation}
Note that when $\xi_t$ is (generalised) white noise, as in this case, then the differential equation~\eqref{eq:system_derivatives} is equivalent to the integral equation~\eqref{eq:sytem_differentials} if the latter is interpreted in It\^o's sense \cite[Section 10.3]{Arnold1974}. Given the equivalence of the two representations in the framework of generalised stochastic processes, in the remainder of the paper equations~\eqref{eq:sytem_differentials} and \eqref{eq:system_derivatives} are used interchangeably, as convenient, to refer to the same nonlinear stochastic system. We refer the reader to \cite[Chapter 3]{Arnold1974} for a detailed discussion on the relation between Brownian motion and white noise, and the representations~\eqref{eq:sytem_differentials} and \eqref{eq:system_derivatives}.

\section{Stochastic Relative Degree and Normal Form}
\label{section:srd_normal_form}
In this section we introduce the concept of stochastic relative degree and show that a suitable coordinate transformation brings the system into a simpler form, which is convenient for analysis and control.

We first introduce three new operators, which are fundamental to systematically define repeated time derivatives of stochastic processes. The first one, which indicates the second derivative of $h$ along the vector fields $f$ and $g$, is defined as
\begin{equation}
\dL{f}{g}{}h(x)\!=\!g(x)^{\top} \dnx{2}[h] \;f(x)=\sum_{j=1}^{n} g_j(x) \sum_{i=1}^{n}\frac{\partial^2 h}{\partial x_j\partial x_i} f_i(x).
\end{equation}
Similarly to the Lie derivative, we use the notation $\dL{a}{b}{}\dL{f}{g}{}h(x)=b(x)^{\top}  \dnx{2}[\dL{f}{g}{}h]\; a(x),$ and $\dL{f}{g}{k}h(x)=g(x)^{\top} \dnx{2}[\dL{f}{g}{k-1}h]\; f(x)$,
to indicate the reiterated operations. The second operator $\SL{f}{l}{}h$, which we call the \emph{stochastic Lie derivative}, indicates the derivative of $h$ along the drift vector field $f$ and diffusion vector field $l$, namely
\begin{equation}
\SL{f}{l}{}h(\xi_t,x)=\Lie_f h(x) + \Lie_l h(x)\xi_t + \frac{1}{2}\dL{l}{l}{}h(x). 
\end{equation}
The reiterated application of this operator can be defined if the white noise does not appear explicitly. That is, if $\SL{f}{l}{}h(\xi_t,x) = \SL{f}{l}{}h(x)$ is a deterministic expression, we use the notation $\SL{f}{l}{2}h(\xi_t,x) = \SL{f}{l}{}\SL{f}{l}{}h(\xi_t,x)$ and, iteratively, if $\SL{f}{l}{k-1}h(\xi_t,x) = \SL{f}{l}{k-1}h(x)$ is deterministic, $\SL{f}{l}{k}h(\xi_t,x) = \SL{f}{l}{}\SL{f}{l}{k-1}h(\xi_t,x)$, with $\SL{f}{l}{0}h(x) = h(x)$ by definition. Finally, we define the third operator
\begin{equation}
\label{eq:first_order_control_term}
	\A{l}{m}{g}{}h(\xi_t,x) = \Lie_g h(x) + \Lie_m h(x) \xi_t + \dL{l}{m}{}h(x).
\end{equation}

Having defined the operators $\dL{}{}{}$, $\SL{}{}{}$ and $\A{}{}{}{}$, it is easy to see that, by using It\^o's formula, the first derivative of the output of system~\eqref{eq:system_derivatives} is given by
\begin{equation}
\label{eq:first_derivative}
y^{(1)}_t = \SL{f}{l}{}h(\xi_t,x_t) + \A{l}{m}{g}{}h(\xi_t,x_t)u + \frac{1}{2}\dL{m}{m}{}h(x_t)u^2.
\end{equation}

We now define the concept of stochastic relative degree and then point out the rationale of such a definition.

\begin{definition}
\label{definition:stochastic_relative_degree}
	(Stochastic Relative Degree)  System~\eqref{eq:system_derivatives} is said to have \emph{stochastic relative degree} $r$ at a point $\bar x$ if
	\begin{description}
        \item[(ND)] $\Lie_l\SL{f}{l}{k}h(x) = 0$ and $\Lie_m \SL{f}{l}{k}h(x) = 0$ for all $x$ in a neighborhood of $\bar x$ and for all $k\in \{0,...,r-2\}$.
		\item[(CD)] $\Lie_g\SL{f}{l}{k}h(x) +  \dL{l}{m}{}\SL{f}{l}{k}h(x)= 0$, $\Lie_m \SL{f}{l}{k}h(x) = 0$ and $\dL{m}{m}{}\SL{f}{l}{k}h(x) = 0$ for all $x$ in a neighborhood of $\bar x$ and all $k\in \{0,...,r-2\}$.
		\item[(RD)] $\Lie_g\SL{f}{l}{r-1}h(\bar x) +  \dL{l}{m}{}\SL{f}{l}{r-1}h(\bar x)\ne 0$ or $\Lie_m \SL{f}{l}{r-1}h(\bar x) \ne 0$ or $\dL{m}{m}{}\SL{f}{l}{r-1}h(\bar x) \ne 0$.
	\end{description}
\end{definition}
As it is formally proved in the next proposition, in Definition~\ref{definition:stochastic_relative_degree}, condition (ND) (which stands for noise decoupling) ensures that the noise $\xi_t$ does not appear in $y_t$ and its first $r-1$ derivatives. In the remainder we will omit the dependency of the operators $\SL{}{}{}$ and $\A{}{}{}{}$ on the white noise $\xi_t$ whenever this does not appear explicitly (for instance because of condition (ND)). Condition (CD) (for control decoupling) ensures that the control input $u$ does not appear in $y_t$ and its first $r-1$ derivatives. Finally, condition (RD) (for relative degree) ensures that the control $u$ appears in the $r$-th derivative of $y_t$, thus defining the relative degree of the system. These observations are formally gathered in the following result, the proof of which can be found in Appendix A.

\begin{proposition}
\label{proposition:stochastic_relative_degree}
Suppose that system~\eqref{eq:system_derivatives} has stochastic relative degree $r>0$ at $\bar x$. Then
\begin{equation}
\label{eq:kth_derivative}
	y^{(k)}_t = \SL{f}{l}{k}h(x_t) \qquad \forall k \in \{0,\dots,r-1\},
\end{equation}
\emph{i.e.} the first $r-1$ derivatives of $y_t$ do not depend explicitly on $\xi_t$  nor $u$. Moreover, if, at time $\bar{t}$, $x_{\bar t} =  \bar x$, then
\begin{multline}
\label{eq:r-th_derivative}
y^{(r)}_{t=\bar t}  = \SL{f}{l}{r}h(\xi_{\bar t},\bar x) + \A{l}{m}{g}{}\SL{f}{l}{r-1}h(\xi_{\bar t},\bar x)u(\bar t)+ \\ \frac{1}{2}\dL{m}{m}{}\SL{f}{l}{r-1}h(\bar x)u(\bar t)^2,
\end{multline}
where either $\A{l}{m}{g}{}\SL{f}{l}{r-1}h(\xi_{\bar t},\bar x)$ or $\dL{m}{m}{}\SL{f}{l}{r-1}h(\bar x)$ are nonzero. 
\end{proposition}

 The previous result shows that, analogously to the deterministic case, the stochastic relative degree is equal to the order of the derivative of the output at time $\bar t$ in which the input $u(\bar t)$ explicitly appears. Two observations are in order: first, while the white noise does not appear in all the derivatives up to order $r-1$ because of condition (ND), it may or may not appear in the $r$-th derivative; second, differently from the deterministic case, the control $u$ appears linearly and/or quadratically in~\eqref{eq:r-th_derivative}.

\begin{remark}
\label{remark:standing_assumption}
     Condition (ND) is a type of \emph{disturbance decoupling} condition, as we suppose that the noise does not appear in $y_t$ and its successive $r-1$ derivatives. For a deterministic analogous, see \emph{e.g.}, \cite[Section 4.6]{Isidori1995}.  If $\Lie_l\SL{f}{l}{k}h(x) \ne 0$ for a $k<r-1$, the differentiation of $y_t$ up to the $r$-th time would require us to introduce successive derivatives of the white noise, which is theoretically and practically challenging. We exclude this possibility with (ND). 
\end{remark}

\begin{remark}
The equality $\Lie_m \SL{f}{l}{k}h(x) = 0$ appears in both conditions (ND) and (CD). This repetition is unnecessary. However, since $\Lie_m \SL{f}{l}{k}h(x)$ multiplies both $\xi_t$ and $u$ in the expression of the derivative $y_t^{(k+1)}$, we believe that for the sake of clarity it is beneficial to require it to be zero in both the noise and control decoupling conditions.
\end{remark}

\begin{remark}
	Consider a linear stochastic system, \emph{i.e.} system~\eqref{eq:system_derivatives} with $f(x_t) = Ax_t$, $g(x_t) \equiv B$, $l(x_t) = Fx_t$ and $m(x_t) \equiv G$, where $A\in\mathbb{R}^{n\times n}$, $B\in\mathbb{R}^{n\times 1}$, $F\in\mathbb{R}^{n\times n}$ and $G\in\mathbb{R}^{n\times 1}$. By applying Definition~\ref{definition:stochastic_relative_degree}, system~\eqref{eq:system_derivatives} in the linear case has stochastic relative degree $r$ if 
	\begin{enumerate}
		\item $CA^kB = 0$ and  $CA^kG = 0$ for all $k\in \{0,...,r-2\}$.
		\item $CA^{r-1}B \ne 0$ or $CA^{r-1}G \ne 0$.
	\end{enumerate} Note that the conditions are remarkably simple and reminiscent of the deterministic case. In fact, $\SL{f}{l}{k}h(x) = CA^k x$, \emph{i.e.} it is linear in $x$, for all $k\in \{0,...,r-1\}$, hence $\dL{l}{l}{}\SL{f}{l}{k}h(x) \equiv \dL{l}{m}{}\SL{f}{l}{k}h(x) \equiv \dL{m}{m}{}\SL{f}{l}{k}h(x) \equiv 0$ because the Hessian of a linear function with respect to $x$ is identically zero. Assuming that the relative degree of this system is, for instance, $r>2$, it follows that
	\begin{equation}
	\begin{aligned}
		y^{(k)}_t &= CA^{k}x_t, \qquad \qquad  \forall k\in \{0,...,r-1\}, \\
	    y^{(r)}_t &= CA^{r-1}(A + F\xi_t)x_t + CA^{r-1}(B+G\xi_t)u.
	\end{aligned}
	\end{equation}
\end{remark}

Having defined the notion of stochastic relative degree and having discussed its interpreation, we are now interested in finding a diffeomorphism $\Phi:\mathbb{R}^n \rightarrow \mathbb{R}^n$ that locally (\emph{i.e.} in a neighbourhood $\bar{U}$ of $\bar x\in U\subset \mathbb{R}^n$) transforms system~\eqref{eq:system_derivatives} in such a way that its dynamics is somewhat ``simpler". The diffeomorphism we are looking for is a direct consequence of the definition of the stochastic relative degree previously given. To simplify the exposition, we make the following assumption on the stochastic Lie derivatives of $y_t = h(x_t)$ along the drift vector $f$ and the diffusion vector $l$.
\begin{assumption}
	\label{assumption:independent_gradients}
	Let $r$ be the stochastic relative degree of system~\eqref{eq:system_derivatives} at $\bar x$. Then the row vectors
	\begin{equation}
		\dnx{} [h]_{x=\bar x}, \; \dnx{} [\SL{f}{l}{}h]_{x=\bar x},\; ... \;,\; \dnx{} [\SL{f}{l}{r-1}h]_{x=\bar x},
	\end{equation}
	are linearly independent.
\end{assumption}
Observe that if Assumption~\ref{assumption:independent_gradients} holds, then necessarily $r\le n$.

\begin{remark}
	For deterministic nonlinear systems, the linear independence of the gradients of the first $r-1$ successive derivatives of the output at $\bar x$ is a consequence of the relative degree being defined, see \emph{e.g}, \cite[Lemma 4.1.1]{Isidori1995}. In Section~\ref{subsection:sharper_results} we prove that this property holds in the present setting for the case in which $m\equiv 0$ near $\bar x$. However, a proof that a similar result holds in general ($m\ne 0$) is missing. For simplicity, we use Assumption~\ref{assumption:independent_gradients} at this stage to develop the theory of normal form for the most general class of systems, and we note that no counter-example has been found for which this property is not satisfied when $m\ne 0$ near $\bar x$.  
\end{remark}

\begin{proposition}
	\label{proposition:general_normal_form}
Suppose that system~\eqref{eq:system_derivatives} has stochastic relative degree $r$ at $\bar{x}$ and let Assumption~\ref{assumption:independent_gradients} hold. Set
\begin{equation}
	\phi_1(x)\!\! =\! \!h(x),\;\;
	\phi_2(x)\!\! =\! \!\SL{f}{l}{}h(x),\;\; \dots, \;\;
	\phi_r(x)\!\! =\!\! \SL{f}{l}{r-1}h(x).
\end{equation}
If $r<n$, then there exist smooth functions $\phi_{r+1}(x),...,\phi_n(x)$, with $\phi_j \in \mathbb{R}$ for all $j \in \{r+1,..,n\}$, such that the Jacobian of the mapping
\begin{equation}
\label{eq:diffeomorphism_general}
	\Phi(x) = \begin{bmatrix}
	\phi_1(x) & \phi_2(x) & \dots & \phi_n(x)
	\end{bmatrix}^\top
\end{equation}
is invertible at $\bar x$ almost surely, thus defining a coordinate transformation in a neighbourhood of $\bar x$. Then the state-space representation of system~\eqref{eq:system_derivatives} in the transformed state $z_t = \Phi(x_t)$ is
\begin{equation}
\begin{aligned}
	\dot{z}_i &= z_{i+1}, &i = 1,...,r-1,\\
	\dot{z}_r &= c(\xi_t,z_t) + b(\xi_t,z_t)u + a(z_t)u^2,	\\
	\dot{z}_{j} &= p_{j}(\xi_t,z_t) + q_{j}(\xi_t,z_t)u + s_{j}(z_t)u^2, \; &j = r+1,...,n,\\
	y_t &= z_1,
\end{aligned}
\label{eq:transformed_system}
\end{equation}
where the mappings\footnote{To keep the statement of the proposition concise, the mappings $a$, $b$, $c$, $p_j$, $q_j$, $s_j$ are defined in the proof.} $c$, $b$, $p_j$ and $q_j$ are affine in $\xi_t$.
\end{proposition}
\begin{proof}
By Assumption~\ref{assumption:independent_gradients}, the matrix
\begin{equation}
	\begin{bmatrix}
	\dnx{} [h(x)]^\top & \dots & \dnx{} [\SL{f}{l}{r-1}h(x)]^\top
	\end{bmatrix}^\top
\end{equation}
has rank $r$ at $\bar{x}$. If $r<n$, let $\gamma_{r+1} (x),...,\gamma_{n} (x)$, with $\gamma_j \in \mathbb{R}^n$ for all $j \in \{r+1,..,n\}$, be any set of $n-r$ vectors such that
\begin{equation}
	\begin{bmatrix}
	\dnx{} [h(\bar{x})]^\top  \dots \; \dnx{} [\SL{f}{l}{r-1}h(\bar{x})]^\top\;\; \gamma_{r+1}(\bar{x}) \;  \dots \; \gamma_{n}(\bar{x})
	\end{bmatrix}^\top
\end{equation}
has rank $n$. Note that this is possible because there always exist $n-r$ linearly independent vectors $\gamma_{r+1} (\bar x),...,\gamma_{n} (\bar x)$ that complete the first $r$ linearly independent vectors $\dnx{} [h(\bar{x})]^\top,...,\dnx{} [\SL{f}{l}{r-1}h(x)]^\top$ to a basis of $\mathbb{R}^n$. Let $\phi_j(x)$ be any smooth function such that $\dnx{} [\phi_j(x)] = \gamma_j^\top(x)$ for $j=r+1,...,n$. Then $\Phi(x)$ as defined in~\eqref{eq:diffeomorphism_general} is a local diffeomorphism in a neighbourhood of $\bar x$ and, therefore, it defines a local change of coordinates $z_t = \Phi(x_t)$ for the stochastic system~\eqref{eq:system_derivatives}. Applying It\^o's lemma and since the system has relative degree $r$, the following holds
\begin{equation}
	\dot{z}_i = \SL{f}{l}{i}h(x_t) = \phi_{i+1}(x_t) = z_{i+1},\qquad i=1,...,r-1.
\end{equation}
Moreover,
\begin{equation}
	\dot{z}_r =\! \SL{f}{l}{r}h(x_t) + \A{l}{m}{g}{}\SL{f}{l}{r-1}h(x_t)u+  \frac{1}{2}\dL{m}{m}{}\SL{f}{l}{r-1}h(x_t)u^2.
\end{equation}
We now set
\begin{equation}
	\begin{aligned}
	c(\xi_t,z_t) &= \SL{f}{l}{r}h(\xi_t,\Phi^{-1}(z_t)),\\
	b(\xi_t,z_t) &= \A{l}{m}{g}{}\SL{f}{l}{r-1}h(\xi_t,\Phi^{-1}(z_t)),\\
	a(z_t) &= \frac{1}{2}\dL{m}{m}{}\SL{f}{l}{r-1}h(\Phi^{-1}(z_t)),
	\end{aligned}
\end{equation}
thus obtaining
\begin{equation}
	\dot{z}_r = c(\xi_t,z_t) + b(\xi_t,z_t)u + a(z_t)u^2.
\end{equation}
As for the remaining $n-r$ components of $z_t$, by applying It\^o's lemma to the functions $\phi_j(x_t)$ and setting
\begin{equation}
	\begin{aligned}
	p_j(\xi_t,z_t) &= \SL{f}{l}{}\phi_j(\xi_t,\Phi^{-1}(z_t)),\\
	q_j(\xi_t,z_t) &= \A{l}{m}{g}{}\phi_j(\xi_t,\Phi^{-1}(z_t)),\\
	s_j(z_t) &= \frac{1}{2}\dL{m}{m}{}\phi_j(\Phi^{-1}(z_t)),
	\end{aligned}
\end{equation}
yields
\begin{equation}
	\dot{z}_{j} = p_{j}(\xi_t,z_t) + q_{j}(\xi_t,z_t)u + s_{j}(z_t)u^2, \qquad j=r+1,...,n.
\end{equation}
The proof is completed by observing that $y_t = h(x_t) = z_1$ and that, by the definitions of the operators $\mathcal{S}$ and $\mathcal{A}$, the coefficients $c$, $b$, $p_j$ and $q_j$ are affine in $\xi_t$.
\end{proof}

Note that it might be possible to find smooth functions $\phi_{r+1},...,\phi_n$ such that the dynamics of the last $n-r$ transformed coordinates is independent of the input $u$, \emph{i.e.} $q_j(\cdot,z_t) \equiv 0$, $s_j(\cdot,z_t)\equiv 0$, for all $j \in \{r+1,...,n\}$, in a neighbourhood of $\Phi(\bar{x})$. This observation motivates the next definition.

\begin{definition}
	(Stochastic Normal Form) Let $x_t$ be the unique solution of~\eqref{eq:system_derivatives} and $z_t = \Phi(x_t)$ be a local diffeomorphism in a subset $U$ of $\mathbb{R}^n$ such that
	\begin{equation}
	\label{eq:normal_form}
	\begin{aligned}
	\dot{z}_i &= z_{i+1},\quad &i=1,...,r-1,\\
	\dot{z}_r &= c(\xi_t,z_t) + b(\xi_t,z_t)u + a(z_t)u^2,\\
	\dot{z}_j &= p_{j}(\xi_t, z_t), &j=r+1,...,n,\\
	y_t &= z_1.
	\end{aligned}
	\end{equation}
	System~\eqref{eq:normal_form} is said to be the \emph{stochastic normal form} of system~\eqref{eq:system_derivatives}.
\end{definition}
\begin{remark}
The fact that the coefficients $c$, $b$, $p_j$ and $q_j$ are affine in $\xi_t$ guarantees that system~\eqref{eq:transformed_system}, which is written in the differential notation (as in~\eqref{eq:system_derivatives}), can always be equivalently written in the integral notation (as in \eqref{eq:sytem_differentials}).
\end{remark}
For compactness, in the remainder we use the definitions  $p=\begin{bmatrix}p_{r+1} & \dots & p_n\end{bmatrix}^\top$, $q=\begin{bmatrix}q_{r+1} & \dots & q_n\end{bmatrix}^\top$ and $s=\begin{bmatrix}s_{r+1} & \dots & s_n\end{bmatrix}^\top$. Obviously, if the stochastic relative degree at $\bar x$ is equal to the order of the system, then the system admits a stochastic normal form in a neighbourhood $U$ of $\bar x$ (because $p$, $q$ and $s$ have dimension zero).
\begin{remark}
	For deterministic systems it can be proved (see, \emph{e.g.}, \cite[Proposition 4.1.3]{Isidori1995}) that functions $\phi_{r+1},...,\phi_n$ always exist such that a normal form exists when $r<n$. While a proof that this property holds in general in the present setting is missing, in Section~\ref{subsection:sharper_results} we prove that this property holds for the case in which $m\equiv0$ near $\bar x$.
\end{remark}

\begin{remark}
The notion of stochastic relative degree and normal form presented are consistent with It\^o's interpretation. This is without loss of generality, as all the results of this section can be obtained also in other formalisms, \emph{e.g.} Stratonovich's formalism \cite{Mellone2020quat}.
\end{remark}

\subsection{Sharper Results for $m\equiv 0$}
\label{subsection:sharper_results}
In the previous section we introduced the concept of relative degree and of normal form for a class of nonlinear stochastic systems in which the control input appears both in the drift and in the diffusion terms of the stochastic differential equation. This allowed us to provide as general as possible definitions. However, for the remainder of the article, it is beneficial to make the standing assumption that the control $u$ does not enter the diffusion term of the stochastic differential equation (\emph{i.e.} $m\equiv 0$) in a neighbourhood of $\bar x$. On the one hand, this sub-class of systems is more common in the literature (see \emph{e.g.} \cite{Pan2001}, \cite{Pan2002} and references therein). On the other hand, this assumption allows us to achieve sharper results in terms of nonlinear control of stochastic systems, both in the idealistic case \emph{i.e.} when the noise process is assumed available) and, in the practically implementable controller that we develop in Section~\ref{subsection:compensating_control}. A discussion of the general case (\emph{i.e.} $m\ne 0$) is given in Section~\ref{sec:noise_in_diffusion}. Therefore, for the time being, we assume $m(x_t) \equiv 0$ near $\bar x$ and we consider systems of the form
\begin{equation}
\label{eq:system_differentials_restricted}
    dx_t = (f(x_t) + g(x_t)u)dt + l(x_t)d\mathcal{W}_t, \qquad y_t =h(x_t),
\end{equation}
or, equivalently,
\begin{equation}
\label{eq:system_derivatives_restricted}
    \dot{x}_t = f(x_t) + g(x_t)u + l(x_t)\xi_t, \qquad y_t =h(x_t).
\end{equation}
In order to prove results in this section it is useful to introduce the \emph{Stratonovich equivalent} of system~\eqref{eq:system_differentials_restricted}, given by
\begin{equation}
\label{eq:system_differentials_restricted_Strat}
    dx_t = (f_S(x_t) + g(x_t)u)dt + l(x_t)\circ d\mathcal{W}_t, \quad y_t =h(x_t),
\end{equation}
where $f_S(x) = f(x) - \frac{1}{2}\frac{\partial l(x)}{\partial x}l(x),$ and the symbol $\circ$ denotes the fact that the stochastic integral is meant in Stratonovich's sense. The It\^o solution $x_t$ of~\eqref{eq:system_differentials_restricted} is identical to the Stratonovich solution of~\eqref{eq:system_differentials_restricted_Strat} (see, \emph{e.g.}, \cite{Wong1965}).

We now discuss what implications follow from restricting ourselves to the class of systems~\eqref{eq:system_derivatives_restricted}. Firstly, the operator $\A{l}{m}{g}{}$, reduces to $\Lie_g$, the Lie derivative along the only control vector field $g$. Secondly, as the control input does not appear in the diffusion term, the derivative of any function of the state $x_t$ is never quadratic in the control input $u$. In fact one can observe that $u^2$ multiplies the operator $\dL{m}{m}{}$, \emph{e.g.} in~\eqref{eq:first_derivative}, and that this operator is identically zero near $\bar x$ because so is the vector field $m$. Having said this, the definition of stochastic relative degree becomes simpler and for the sake of clarity it is useful to rewrite it.

\begin{definition}
\label{definition:stochastic_relative_degree_restricted}
	(Stochastic Relative Degree - $m \equiv 0$) System~\eqref{eq:system_derivatives_restricted} is said to have \emph{stochastic relative degree} $r$ at a point $\bar x$ if
	\begin{description}
        \item[(ND)] $\Lie_l\SL{f}{l}{k}h(x) = 0$ for all $x$ in a neighborhood of $\bar x$ and for all $k\in \{0,...,r-2\}$.
		\item[(CD)] $\Lie_g\SL{f}{l}{k}h(x) = 0$ for all $x$ in a neighborhood of $\bar x$ and all $k\in \{0,...,r-2\}$.
		\item[(RD)] $\Lie_g\SL{f}{l}{r-1}h(\bar x) \ne 0$.
	\end{description}
\end{definition}
By considering a class of systems of the form~\eqref{eq:system_derivatives_restricted}, we can achieve sharper results in terms of definition of a stochastic normal form. Indeed, we can now prove the claim previously assumed in Assumption~\ref{assumption:independent_gradients}. To this end, we first prove a technical result.

\begin{lemma}
\label{lemma:identity_of_lie_derivatives}
Let $r$ be the stochastic relative degree of system~\eqref{eq:system_derivatives_restricted} at $\bar x$. Then
\begin{equation}
\label{eq:identity_stochastic_lie_derivatives}
    \SL{f}{l}{k}h = \Lie_{f_S}^k h \qquad \forall k \in \{0,...,r-1\} \quad \text{near $\bar x$.}
\end{equation}
\end{lemma}
\begin{proof}
First observe that by (ND),
\begin{equation}
    \partial_x\left(\Lie_l\SL{f}{l}{k}h\right)= l^\top\frac{\partial^2 \SL{f}{l}{k}h}{\partial x^2} + \frac{\partial \SL{f}{l}{k}h}{\partial x}\frac{\partial l}{\partial x} = 0,
\end{equation}
hence
\begin{equation}
\label{eq:quadratic_terms_identity}
    l^\top\frac{\partial^2 \SL{f}{l}{k}h}{\partial x^2} = - \frac{\partial \SL{f}{l}{k}h}{\partial x}\frac{\partial l}{\partial x},
\end{equation}
for all $k$ in $\{0,...,r-2\}$ near $\bar x$. Then the claim follows by induction. In fact, \eqref{eq:identity_stochastic_lie_derivatives} trivially holds for $k=0$. Note that for some $0<k<r-1$ we have, in view of~\eqref{eq:quadratic_terms_identity}, that
\begin{align}
    \SL{f}{l}{k+1}h\! &=\! \Lie_f \!\!\SL{f}{l}{k}h + \frac{1}{2}l^\top \frac{\partial^2 \SL{f}{l}{k}h}{\partial x^2}l\\ &= \Lie_f\!\! \SL{f}{l}{k}h - \frac{1}{2}\frac{\partial \SL{f}{l}{k}h}{\partial x}\frac{\partial l}{\partial x}l
    = \Lie_{f - \frac{1}{2}\frac{\partial l}{\partial x}l}\SL{f}{l}{k}h
\end{align} Rewriting this equation using in the Lie derivative notation and assuming  by the inductive hypothesis that~\eqref{eq:identity_stochastic_lie_derivatives} holds for $0<k<r-1$ yields
\begin{equation}
    \SL{f}{l}{k+1}h = \Lie_{f-\frac{1}{2}\frac{\partial l}{\partial x}l}\Lie_{f_S}^k h = \Lie_{f_S}\Lie_{f_S}^k h = \Lie_{f_S}^{k+1} h.
\end{equation}
\end{proof}
The fact that the stochastic Lie derivatives of $h$ along $f$ and $l$ are identical to the Lie derivatives of $h$ along $f_S$ up to order $r-1$ is crucial to prove the rest of the results in this and the next section. In fact, this equivalence allows us to leverage the deterministic techniques, because the Stratonovich differentiation rule is formally identical to the deterministic one. We can then use the identities in Lemma~\ref{lemma:identity_of_lie_derivatives} to translate the results to the It\^o system~\eqref{eq:system_derivatives_restricted}.

We are now ready to prove the following.
\begin{lemma}
\label{lemma:gradients_independent}
Let $r$ be the stochastic relative degree of system~\eqref{eq:system_derivatives_restricted} at $\bar x$. Then the row vectors
	\begin{equation}
		\dnx{} [h]_{x=\bar x}, \; \dnx{} [\SL{f}{l}{}h]_{x=\bar x},\; ... \;,\; \dnx{} [\SL{f}{l}{r-1}h]_{x=\bar x},
	\end{equation}
	are linearly independent.
\end{lemma}
\begin{proof}
Lemma~\ref{lemma:identity_of_lie_derivatives} and the definition of stochastic relative degree of~\eqref{eq:system_derivatives_restricted} imply that the set of conditions $\Lie_g \Lie_{f_S}^k h = 0$ near $\bar x$ for all $k$ in $\{0,...,r-2\}$ and $\Lie_g \Lie_{f_S}^{r-1} h \ne 0$ at $\bar x$ hold. By \cite[Lemma 4.1.1]{Isidori1995}, this set of conditions in turn implies the linear independence of
    \begin{equation}
		\dnx{} [h]_{x=\bar x}, \; \dnx{} [\Lie_{f_S}h]_{x=\bar x},\; ... \;,\; \dnx{} [\Lie_{f_S}^{r-1}h]_{x=\bar x}.
	\end{equation}
	Using again Lemma~\ref{lemma:identity_of_lie_derivatives}, the claim follows.
\end{proof}

For systems of the form~\eqref{eq:system_derivatives_restricted} the result in Proposition~\ref{proposition:general_normal_form} can be specialised to the existence of a coordinate transformation $z_t = \Phi(x_t)$ such that the transformed dynamics is given by
\begin{equation}
\begin{aligned}
	\dot{z}_i &= z_{i+1}, &i = 1,...,r-1,\\
	\dot{z}_r &= c(\xi_t,z_t) + b(z_t)u,	\\
	\dot{z}_{j} &= p_{j}(\xi_t,z_t) + q_{j}(z_t)u, \; &j = r+1,...,n,
\end{aligned}
\label{eq:transformed_system_restricted}
\end{equation}
where the coefficients $c$ and $p_j$ are as before and
\begin{equation}
 	b(z_t)\! =\! \Lie_g\!\SL{f}{l}{r-1}h(\Phi^{-1}(z_t)), \quad
 	q_j(z_t) \!= \!\Lie_g\!\phi_j(\Phi^{-1}(z_t)).
\end{equation}
Observe that, unlike the general case in which $m\ne0$, since now the coefficients $q_j$ are simply the Lie derivatives of $\phi_j$ along $g$ for $j=r+1,\dots,n$, it is easy to show that it is always possible to find $\phi_{r+1}, \dots , \phi_{n}$ such that said coefficients are identically zero for $x$ in a neighbourhood of $\bar x$. Exploiting Lemma~\ref{lemma:identity_of_lie_derivatives}, the proof is analogous to the one reported in \cite[Proposition 4.1.3]{Isidori1995}. Thus, we can conclude that it is always possible to find a coordinate transformation $\Phi(x)$ in a subset $U$ of $\mathbb{R}^n$ such that the dynamics in the transformed state $z_t = \Phi(x_t)$ is
\begin{equation}
	\label{eq:normal_form_restricted}
	\begin{aligned}
	\dot{z}_i &= z_{i+1},\quad &i=1,...,r-1,\\
	\dot{z}_r &= c(\xi_t,z_t) + b(z_t)u,\\
	\dot{z}_j &= p_{j}(\xi_t, z_t), &j=r+1,...,n,\\
	y_t &= z_1,
	\end{aligned}
	\end{equation}
which is the normal form of system~\eqref{eq:system_derivatives_restricted}.

\begin{remark}
The notion of normal form which we propose introduces elements of significant novelty, compared to works on analogous topics like \cite{Ladhiri1996}, \cite{Pan2001} and \cite{Pan2002}. By condition (ND) in Definition~\ref{definition:stochastic_relative_degree_restricted}, the state is projected onto new coordinates, of which the first $r-1$ have a noise-free dynamics. Not only is this a fundamental difference with respect to past works, but it also allows us to introduce, as shown in Sections~\ref{section:linearisation} and \ref{sec:practical_linearisation}, control laws performing path-wise control of nonlinear stochastic systems. To the best of the authors' knowledge, this problem has not been systematically addressed in the literature.
\end{remark}

\section{Idealistic Linearisation via State Feedback}
\label{section:linearisation}
We now turn our attention to the problem of exact feedback linearisation. Specifically, in this section we formulate and solve an idealistic version of the problem, by assuming that the control input is a function of the noise $\xi_t$. Although unrealistic because it implies that the noise is available for feedback, this assumption is necessary to develop the theory of \emph{exact} linearisation via feedback. This theory is instrumental for the introduction of a practically implementable controller in Section~\ref{sec:practical_linearisation} which, by approximating the noise contribution to the dynamics via measurements of the state, achieves a \emph{practical} result. The feedback linearisation problem is formally defined as follows.

\begin{problem}
	\label{problem:exact_linearisation_via_state_feedback}
	(Exact Feedback Linearisation) Consider the nonlinear stochastic system
	\begin{equation}
	\label{eq:general_system_no_output}
		\dot{x}_t = f(x_t) + g(x_t)u + l(x_t)\xi_t.
	\end{equation}
	Given a point $\bar x$, the problem of \emph{exact feedback linearisation} consists in finding a neighbourhood $U$ of $\bar x$, a feedback law $u_t = k(\xi_t,x_t,v)$ affine in $\xi_t$, with $v\in \mathbb{R}$, defined on $U$ and a stochastic coordinate transformation $z_t = \Phi(x_t)$ defined on $U$ such that the closed-loop system
	\begin{equation}
		\dot{x}_t = f(x_t) + g(x_t)k(\xi_t,x_t,v) + l(x_t)\xi_t,
	\end{equation}
	in the coordinates $z_t = \Phi(x_t)$, is linear, deterministic and controllable.
\end{problem}

\begin{remark}
By requiring that the control $u_t = k(\xi_t,x_t,v)$ is an explicit function of the white noise, as in the statement of Problem~\ref{problem:exact_linearisation_via_state_feedback}, we may be enlarging the class of systems~\eqref{eq:system_derivatives_restricted} beyond what could be allowed. In the general case where $k$ is any nonlinear function of $\xi_t$, the resulting closed-loop system in the derivative notation, \emph{i.e.} \eqref{eq:system_differentials_restricted}, is not anymore equivalent to the one in the differential notation, \emph{i.e.} \eqref{eq:system_derivatives_restricted}. In fact, this equivalence is preserved only when the resulting closed-loop dynamics is affine in the white noise. Therefore, we hereby define an idealistic control law to be \emph{admissible} if the dynamics of the closed-loop system is affine in $\xi_t$. This is achieved easily by requiring that $k$ is affine in $\xi_t$ as done in the statement of Problem~\ref{problem:exact_linearisation_via_state_feedback}. Note that the idealistic controls that we introduce in this paper are all admissible by construction.
\end{remark}

\begin{proposition}
	\label{proposition:feedback_linearisation}
	Problem~\ref{problem:exact_linearisation_via_state_feedback} is solvable if and only if there exists a real valued function $h(x_t)$ such that system~\eqref{eq:general_system_no_output} with the output $y_t = h(x_t)$ has stochastic relative degree $n$ at $\bar x$.
\end{proposition}

\begin{proof}
    \emph{Sufficiency.} Suppose $y_t=h(x_t)$ is such that the stochastic relative degree of system~\eqref{eq:general_system_no_output} is $n$ at $\bar x$. Then by Proposition~\ref{proposition:general_normal_form}, there exists a change of coordinates
	\begin{equation}
		z_t = \Phi(x_t) = \begin{bmatrix}
		h(x_t) & \SL{f}{l}{}h(x_t) & \dots & \SL{f}{l}{n-1}h(x_t)
		\end{bmatrix}^\top,
	\end{equation}
	for all $x$ in a neighbourhood $U$ of $\bar x$, such that the normal form of system~\eqref{eq:general_system_no_output} is
	\begin{equation}
	\begin{aligned}
		\dot{z}_i &= z_{i+1},\quad i = 1,...,n-1,\\
		\dot{z}_n &= c(\xi_t,z_t) + b(z_t)u.
	\end{aligned}
	\end{equation} By the definition of stochastic relative degree, there exists a neighbourhood $U$ of $\bar{x}$ such that $b(z_t) \ne 0$ in $\Phi(U)$. Therefore the control law
	\begin{equation}
	\label{eq:linearising_feedback_first_order}
		u_t = \tilde{k}(\xi_t,z_t, v) = \frac{1}{b(z_t)}(-c(\xi_t,z_t) + v)
	\end{equation}
	is well-defined in $\Phi(U)$. Therefore, in the neighbourhood $U$ of $\bar x$ the feedback law $u_t = \tilde{k}(\xi_t,\Phi(x_t),v)= k(\xi_t,x_t,v)$ defined on $U$ brings the transformed system with state $z_t = \Phi(x_t)$ into the form
	\begin{equation}
	\label{eq:companion_form}
		\dot{z}_t = \begin{bmatrix}
		0 & 1 & 0 & \dots & 0\\
		0 & 0 & 1 & \dots & 0\\
		\vdots & \vdots & \vdots & \vdots & \vdots\\
		0 & 0 & 0 & \dots & 1\\
		0 & 0 & 0 & \dots & 0
		\end{bmatrix}z_t + \begin{bmatrix}
		0 \\ 0 \\ \vdots \\ 0 \\ 1
		\end{bmatrix}v = Az_t + Bv,
	\end{equation}
	where $A$ and $B$ are such that $\begin{bmatrix}B & AB & \dots & A^{n-1}B\end{bmatrix}$ has rank $n$. The transformed system is linear, deterministic and controllable, hence Problem~\ref{problem:exact_linearisation_via_state_feedback} is solved.\\
	\emph{Necessity.} This proof is formally identical to the one of \cite[Lemma 4.2.1]{Isidori1995}. In fact, it is possible to show that the stochastic relative degree is invariant under coordinate change and feedback by using the arguments in \cite{Isidori1995} on the Stratonovich equivalent system~\eqref{eq:system_differentials_restricted_Strat} and then use Lemma~\ref{lemma:identity_of_lie_derivatives} to show that these arguments hold for the It\^o system~\eqref{eq:system_derivatives_restricted} as well.
\end{proof}

\begin{remark}
	As the coefficient $c$ is in general a function of $\xi_t$, the feedback linearising control~\eqref{eq:linearising_feedback_first_order} could require the knowledge of the exact value of the white noise for all $t$, because this is the only way that the noisy dynamics of the open loop system~\eqref{eq:general_system_no_output} can be rendered deterministic via feedback. Of course, the assumption that the noise can be used in the feedback loop is unrealistic. However, building on the results of this section, in Section~\ref{sec:practical_linearisation} we design a practical hybrid control scheme. This hybrid scheme implements a deterministic state feedback law which periodically  compensates, in an approximate way, the noise.
\end{remark}

\begin{remark}
The control $u_t= k(\xi_t, x_t, v)$ designed in the proof of Proposition~\ref{proposition:feedback_linearisation} is affine in the variable $\xi_t$, because in~\eqref{eq:linearising_feedback_first_order} the coefficient $b$ is not an explicit function of $\xi_t$, whilst the coefficient $c$ is affine in $\xi_t$. Therefore, the feedback linearising control is always admissible because replacing its expression in~\eqref{eq:system_derivatives_restricted} leaves the closed-loop dynamics affine in $\xi_t$.
\end{remark}

To conclude this section, we provide necessary and sufficient conditions for the existence of an output $y_t = h(x_t)$ which makes the stochastic relative degree of system~\eqref{eq:general_system_no_output} equal to $n$ at a point $\bar x$.
\begin{theorem}
Problem~\ref{problem:exact_linearisation_via_state_feedback} is solvable if and only if
\begin{enumerate}
	\item the matrix $\begin{bmatrix}
	g(\bar x) & \text{ad}_{f_S}g(\bar x) & \dots & \text{ad}_{f_S}^{n-1}g(\bar x)
	\end{bmatrix}$ is invertible and
	\item the distribution $\text{span}\{g, \text{ad}_{f_S}g, \dots, \text{ad}_{f_S}^{n-2}g\}$ is involutive near $\bar x$.
\end{enumerate}
\end{theorem}
\begin{proof}
By Proposition~\ref{proposition:feedback_linearisation}, Problem~\ref{problem:exact_linearisation_via_state_feedback} is solvable if and only if it is possible to find an output function $h(x_t)$ satisfying the conditions in Definition~\ref{definition:stochastic_relative_degree_restricted}. By Lemma~\ref{lemma:identity_of_lie_derivatives}, this amounts to the set of conditions
\begin{enumerate}
		\item $\Lie_g\Lie_{f_S}^{k}h(x) = 0$ for all $x$ in a neighborhood of $\bar x$ and all $k\in \{0,...,r-2\}$.
		\item $\Lie_g\Lie_{f_S}^{r-1}h(\bar x) \ne 0$,
\end{enumerate}
which is a set of partial differential equations and a non-triviality condition. By \cite[Lemma 4.1.2]{Isidori1995} these are equivalent to
\begin{enumerate}
	\item $\Lie_g\Lie_{\text{ad}_{f_S}g}^{k}h(x) = 0$ for all $x$ in a neighborhood of $\bar x$ and all $k\in \{0,...,r-2\}$.
	\item $\Lie_g\Lie_{\text{ad}_{f_S}g}^{r-1}h(\bar x) \ne 0$.
\end{enumerate}
Following the proof of \cite[Lemma 4.2.2]{Isidori1995} it is possible to show that the previous set of conditions is equivalent to the set of conditions in the statement of this theorem.
\end{proof}

If the stochastic relative degree $r$ of system~\eqref{eq:system_derivatives_restricted} is strictly less than $n$ at $\bar x$, we indicate by $\zeta_t = \begin{bmatrix}z_1 & \dots & z_r\end{bmatrix}^\top$ the first $r$ components of the transformed state $z_t = \Phi(x_t)$ and by $\eta_t = \begin{bmatrix}z_{r+1} & \dots & z_n\end{bmatrix}^\top$ the remaining $n-r$. Note that if $r< n$ the control law~\eqref{eq:linearising_feedback_first_order} linearises only the dynamics of $\zeta_t$, thus we say that it partially feedback linearises system~\eqref{eq:normal_form_restricted}.

The zero dynamics of a nonlinear stochastic system is, analogously to the deterministic case, the dynamics of the internal variable $\eta_t$ when the input and the initial conditions are chosen in such a way that the output is constrained to be identically zero. This is achieved by setting $\zeta_0 = 0$ and $u_{z,t} = -c(\xi_t,0,\eta_t)/b(0,\eta_t)$. Note that the control $u_{z,t}$ is affine in $\xi_t$, hence admissible. We now provide a definition that extends the concept of zero dynamics to nonlinear stochastic systems of the form~\eqref{eq:system_derivatives_restricted}.

\begin{definition}
\label{definition:zero_dynamics}
(\emph{Zero Dynamics}) The stochastic differential equation
\begin{equation}
    \dot{\eta}_t = p(\xi_t,0,\eta_t)
\end{equation}
is called the \emph{zero dynamics} of system~\eqref{eq:system_derivatives_restricted}.
\end{definition}

The properties of the zero dynamics are fundamental in studying the problem of asymptotic output tracking, which is the topic of Section~\ref{sec:stab_track}.

\section{Practical Linearisation via State Feedback}
\label{sec:practical_linearisation}
In this section we introduce a causal method to estimate the sequence of variations of the Brownian motion between successive time instants. We then show that the use of these estimates is beneficial in the design of practical feedback linearising controls.

\subsection{Estimation of the Brownian motion}
\label{sec:noise_estimation}
The estimation of the Brownian variations is performed by periodically sampling the state. This method was first introduced in \cite{Mellone2021} in order to practically solve the problem of output regulation of \emph{linear} stochastic systems. We now extend it to the present context of \emph{nonlinear} stochastic systems.

Let $\{t_k\}_{k\in\mathbb{Z}_{\ge 0}}$ be a sequence of equally-spaced sampling times, with $t_k - t_{k-1} = \varepsilon$ for all $k\in\mathbb{Z}_{>0}$. Define the differences $\Delta W_\varepsilon(k) = \W_{t_k} - \W_{t_{k-1}}$ and $\Delta x(k) = x_{t_k} - x_{t_{k-1}}$. Our aim is to show that it is possible to compute a causal estimate $\Delta \widehat{W}_\varepsilon(k)$ of the quantity $\Delta W_\varepsilon(k)$ by comparing the samples of the state of the system at times $t_{k-1}$ and $t_k$. In particular, we want this estimate to ``converge'', in a sense to be defined, to the stochastic differential $d\W_t$ as the sampling period $\varepsilon$ converges to zero. Let $\mathfrak{L}_I$ be the space of functions that are integrable in It\^o's sense. Then with the notation
$ \Delta \widehat{W}_\varepsilon \xrightarrow{\varepsilon} d\W_t$ we mean that for all $\alpha\in\mathfrak{L}_I$
$\lim_{\varepsilon \rightarrow 0} \sum_{k}\alpha(t_{k-1},w)\Delta \widehat{W}_\varepsilon(k) = \int_0^t\alpha(\tau,w)d\mathcal{W}_\tau$ for almost all $w\in \nabla$.
For ease of notation, define
$F_{t_k} = f(x_{t_k}) + g(x_{t_k})u_{t_k}$ and $L_{t_k} = l(x_{t_k}),$
which are the drift and diffusion coefficients, respectively, of system~\eqref{eq:system_differentials_restricted} evaluated at time $t_k$. We make the following assumption.
\begin{assumption}
\label{assumption:persistence_excitation}
There exists $\delta>0$ such that $\vert L_{t_k} \vert > \delta$ almost surely for all $k \in \mathbb{Z}_{\ge 0}$.
\end{assumption}
The rationale of this assumption is explained later, in Remark~\ref{remark:persistence_excitation} at the end of Section~\ref{sec:practical_linearisation}. Under Assumption~\ref{assumption:persistence_excitation} the Moore-Penrose left pseudo-inverse of $L_{t_k}$, \emph{i.e.} $L_{t_k}^\dagger = (L_{t_k}^\top L_{t_k})^{-1}L_{t_k}^\top$, is well-defined almost surely. The following Lemma extends the results in \cite{Mellone2021} to systems with nonlinear drift and diffusion terms.

\begin{lemma}
\label{lemma:brownian_motion_approximation}
Consider system~\eqref{eq:system_derivatives_restricted} and let Assumption~\ref{assumption:persistence_excitation} hold. Let $\{\Delta \widehat{W}_\varepsilon(k)\}_{k>0}$ be a sequence of scalars defined as
\begin{equation}
\label{eq:dWest_definition}
    \Delta \widehat{W}_\varepsilon(k) = L_{t_{k-1}}^\dagger\left[ \Delta x(k) - F_{t_{k-1}}\varepsilon \right].
\end{equation}
Then $\Delta \widehat{W}_\varepsilon(k)\xrightarrow{\varepsilon}d\W_t$ almost surely.
\end{lemma}
\begin{proof}
Let $k\in\mathbb{Z}_{>0}$. By \cite[Theorem 7.1]{Gard1988}
	\begin{equation}
		\Delta x(k) =  F_{t_{k-1}}\varepsilon + L_{t_{k-1}}\Delta W_\varepsilon(k) + o(\varepsilon^2),
	\end{equation}
	holds, where $o(\varepsilon^2)$, which is the \textit{one-step truncation error} of the forward-Euler scheme, is an infinitesimal of the same order of $\varepsilon^2$. The previous expression can be rewritten as
	\begin{equation}
		L_{t_{k-1}}\Delta W_\varepsilon(k) = \Delta x(k) - F_{t_{k-1}}\varepsilon + o(\varepsilon^2).
	\end{equation}
	Since $L_{t_{k-1}}$ has full column rank almost surely, the expression
	\begin{equation}
		\Delta W_\varepsilon(k) = L_{t_{k-1}}^\dagger[\Delta x(k) -F_{t_{k-1}}\varepsilon + o(\varepsilon^2)],
	\end{equation}
	holds almost surely. Defining $\Delta \widehat{W}_\varepsilon(k)$ as in~\eqref{eq:dWest_definition} yields
	\begin{equation}
	\label{eq:relation_DeltaW_true_est}
		\Delta \widehat{W}_\varepsilon(k) = \Delta W_\varepsilon(k) + L_{t_{k-1}}^\dagger o(\varepsilon^2).
	\end{equation}
	almost surely. Let $\alpha_t\in \mathfrak{L}_I$. Then
	\begin{equation}
	    \sum_k \alpha_{t_{k-1}} \Delta \widehat{W}_\varepsilon(k) =  \sum_k \alpha_{t_{k-1}} (\Delta W_\varepsilon(k) + L_{t_{k-1}}^\dagger o(\varepsilon^2)).
	\end{equation}
	Taking the limit of both sides as $\varepsilon$ tends to zero yields $\Delta \widehat{W}_\varepsilon \xrightarrow{\varepsilon}d\mathcal{W}_t$, since for all $\alpha\in\mathfrak{L}_I$
	\begin{equation}
	    \lim_{\varepsilon \rightarrow 0} \sum_k \alpha_{t_{k-1}}L_{t_{k-1}}^\dagger o(\varepsilon^2) = 0 \quad \text{almost surely.}
	\end{equation}
\end{proof}

\subsection{Hybrid control law}
\label{subsection:compensating_control}
In this section we discuss how the sequence of estimates $\{\Delta \widehat{W}_\varepsilon(k)\}_k$ can be used to design a practically implementable hybrid feedback control law that approximates the idealistic input~\eqref{eq:linearising_feedback_first_order}, namely $u_t^{lin} = (-c(\xi_t,z_t)+v)/b(z_t)$, which linearises the dynamics of the first $r$ components of system~\eqref{eq:normal_form_restricted}. We show that the accuracy can be improved by increasing the sampling time $\varepsilon$. We formally state the problem as follows.
\begin{problem}
\label{problem:path-wise_control}
    (Practical Feedback Linearisation) Consider system~\eqref{eq:normal_form_restricted} and let $z_t^{lin}$ be its solution when $u_t = u_t^{lin}$. The problem of \emph{practical feedback linearisation} consists in finding a control law $\hat{u}^{lin}_t(\varepsilon)$, depending on the sampling rate $\varepsilon$, such that, if $z_t$ is the solution of \eqref{eq:normal_form_restricted} when $u_t = \hat{u}^{lin}_t$, then for every $\sigma>0$
    \begin{equation}
    \label{eq:convergence_to_ideal_trajectory}
        \lim_{\varepsilon\rightarrow 0} \mathfrak{P}(\left| z_t - z_t^{lin} \right| \ge \sigma) = 0.
    \end{equation}
\end{problem}

The meaning of Problem~\ref{problem:path-wise_control} is to find a causal controller for which the closed-loop system dynamics converges in probability to the idealistic closed-loop dynamics as the sampling time is made smaller. To begin with, since the coefficient $c$ is affine in $\xi_t$, we can express it as $c(\xi_t,z_t) = c_d(z_t) + c_s(z_t)\xi_t$ for some mappings $c_d$ and $c_s$. Note that since in practice $\xi_t$ is not known, a naive approximation of $\xi_t$ boils down to replacing it with its expectation, \emph{i.e.} zero. Therefore, a first causal approximation of the coefficient $c$ is given by only the term $c_d(z_t)$, in turn implying that the control $u_t^{lin}$ can be practically approximated by the naive law
\begin{equation}
\label{eq:approximate_control_no_noise}
u^{zn,lin}_t = \frac{-c_d(z_t) + v}{b(z_t)}.
\end{equation}
We call this basic feedback law the \emph{zero-noise control}. By replacing $c_s(z_t)\xi_t$ with zero in the expression of $c(\xi_t,z_t)$, the zero-noise control does not perform any form of stochastic compensation when performing feedback linearisation. Consequently, there is no guarantee that the closed-loop behaviour of system~\eqref{eq:normal_form_restricted} with $u_t = u_t^{zn,lin}$ is any close to its idealistic behaviour, \emph{i.e.} when $u_t = u_t^{lin}$.

The goal of this section is to improve the performance of the zero-noise control $u_t^{zn,lin}$ by leveraging the estimated sequence $\{\Delta \widehat{W}_\varepsilon(k)\}_k$ and show that we can recover the idealistic behaviour in probability. Therefore, we define $\hat{u}_t^{lin} = u^{zn,lin}_t + u^s_t$, with $u^s_t$ to be specified in such a way that $\hat{u}_t^{lin}$ is a ``better" approximation of $u_t^{lin}$ than $u_t^{zn,lin}$. By replacing the control $\hat{u}_t^{lin}$ in~\eqref{eq:normal_form_restricted}, the dynamics of the transformed system becomes
\begin{equation}
\label{eq:dynamics_compensating_control}
    \begin{aligned}
    \dot{z}_i &= z_{i+1}, \qquad\qquad i = 1,\dots,r-1,\\
    \dot{z}_r &= v +c_s(z_t)\xi_t + b(z_t)u^s_t,\\
    \dot{\eta}_t &= p(\xi_t,\zeta_t,\eta_t),\\
    y_t &= z_1,
    \end{aligned}
\end{equation}
where the term $c_s(z_t)\xi_t$ in the dynamics of the $r$-th component is due to the fact that the approximating control $u^{zn,lin}_t$ cannot cancel the noisy dynamics as the idealistic control $u_t^{lin}$ does.

We now want to design the control $u_t^s$ employing the estimates $\{\Delta \widehat{W}_\varepsilon(k)\}_k$ introduced in Section~\ref{sec:noise_estimation}, to reduce the contribution of the term $c_s(z_t)\xi_t$ onto the dynamics of the system. Since the quantity $\Delta \widehat{W}_\varepsilon(k)$ carries information on the evolution of the noise between $t_{k-1}$ and $t_k$, we look at the evolution of $z_r$ between these two consecutive sampling times. The value of $z_r$ at time $t_k$ is given by
\begin{equation}
    z_{r,t_{k}} = z_{r,t_{k-1}} + \int_{t_{k-1}}^{t_{k}}\!\!\!\!\!\!\!vd\tau +  \beta_{d}(k) + \int_{t_{k-1}}^{t_{k}}\!\!\!\!\!\!\!b(z_\tau)u^s_{\tau}d\tau,
\end{equation}
where
\begin{equation}
    \beta_{d}(k) = \int_{t_{k-1}}^{t_{k}}\!\!\!\!\!\!\!c_s(z_\tau)d\W_\tau
\end{equation}
is the contribution of the noise on the dynamics of $z_r$ between the two sampling times. Our goal is to minimise this contribution using $u^s_t$ and the estimate $\Delta\widehat{W}_\varepsilon(k)$ obtained at time $t_{k}$. The fact that $\Delta\widehat{W}_\varepsilon(k)$ is only available \emph{a posteriori}, at the end of the sampling period, suggests that $u^s_t$ should induce a jump variation at time $t_k$ in the state $z_{r}$ in order to compensate for the quantity $\beta_d(k)$. In other words, the dynamics of the closed-loop system should be hybrid. It is then necessary to introduce a simplified jump notation. At time $t_k$, we denote by $z_{t_{k}}$ the state before the jump, and by $z_{t_{k}^+}$ the state after the jump. The flow dynamics of the closed-loop hybrid system we seek is, therefore, given by
\begin{equation}
\label{eq:flow_dynamics}
    \begin{aligned}
    \dot{z}_i &= z_{i+1}, \qquad\qquad i = 1,\dots,r-1,\\
    \dot{z}_r &= v +c_s(z_t)\xi_t,\\
    \dot{\eta}_t &= p(\xi_t,\zeta_t,\eta_t),
    \end{aligned}
\end{equation}
for all $t\in \mathbb{R}_{\ge 0}$, while the jump dynamics is given by
\begin{equation}
\label{eq:jump_dynamics}
    \begin{aligned}
    z_{i,t_k^+} &= z_{i,t_k}, \qquad\qquad i = 1,\dots,r-1,\\
    z_{r,t_k^+} &= z_{r,t_k} + b(z_{t_k})u^*(k),\\
    \eta_{t_k^+} &= \eta_{t_k},
    \end{aligned}
\end{equation}
for all $k\in \mathbb{Z}_{>0}$, where $\{u^*(k)\}_{k}$ is a yet to be defined sequence of scalars depending on $\{\Delta \widehat{W}_\varepsilon(k)\}_k$. Alternatively, the aforementioned hybrid dynamics \eqref{eq:flow_dynamics}-\eqref{eq:jump_dynamics} can be equivalently produced by using in~\eqref{eq:dynamics_compensating_control} an impulsive control $u_t^{s}$ given by
\begin{equation}
\label{eq:compensating_control}
    u^s_t = \sum_{i=1}^{k}u^*(i)\delta(t-t_{i}), \qquad t\le t_{k},
\end{equation}
where $\delta(t)$ is a Dirac delta. As a result of the hybrid dynamics induced by the control $u_t^s$, the expression of $z_{r}$ before the jump at time $t_k$ is given by
\begin{equation}
   z_{r,t_k} = z_{r,t_{k-1}^+} + \int_{t_{k-1}}^{t_{k}}\!\!\!\!\!\!\!vd\tau +  \beta_{d}(k),
\end{equation}
while after the jump by
\begin{align}
\label{eq:state_after_jump}
    z_{r,t_{k}^+} &= z_{r,t_k} + b(z_{t_{k}})u^*(k)\\
                  &= z_{r,t_{k-1}^+} + \int_{t_{k-1}}^{t_{k}}\!\!\!\!\!\!\!vd\tau +  \beta_{d}(k)  + b(z_{t_{k}})u^*(k).
\end{align}
Thus, we have reduced the problem of approximate partial feedback linearisation to the problem of finding the sequence $\{u^*(k)\}_k$ such that the contribution of the term $\beta_d(k) + b(z_{t_{k}})u^*(k)$ is minimised. In particular, we look for a sequence $\{u^*(k)\}_k$ which retrieves the exact  linearisation of the dynamics of $\zeta_t$ as $\varepsilon$ tends to zero. We are now ready to solve Problem~\ref{problem:path-wise_control}.

\begin{theorem}
\label{theorem:approximate_retrieves_ideal}
Consider system~\eqref{eq:normal_form_restricted} and let Assumption~\ref{assumption:persistence_excitation} hold. Let the control $\hat{u}_t^{lin}$ be given by
\begin{equation}
\label{eq:apprximate_partial_feedback_linearising_control}
    \hat{u}_t^{lin} = \frac{-c_d(z_t) + v}{b(z_t)} - \sum_{i=1}^{k}\frac{c_s(z_{t_{i-1}^+})\Delta \widehat{W}_\varepsilon(i)}{b(z_{t_i})}\delta(t-t_{i}), \quad t\le t_{k},
\end{equation}
with $\Delta \widehat{W}_\varepsilon(k)$ given by~\eqref{eq:dWest_definition}. Then $\hat{u}_t^{lin}$ solves Problem~\ref{problem:path-wise_control}.
\end{theorem}
\begin{proof}
It is trivial to observe that, since system~\eqref{eq:normal_form_restricted} is in normal form, \eqref{eq:convergence_to_ideal_trajectory} holds if and only if it holds for the $r$-th component of the state $z_t$. Thus, we now focus on the $r$-th component. Under the control $\hat{u}_t^{lin}$, the jump of $z_r$ at $t_k$ is given by
\begin{equation}
    z_{r,t_{k}^+} = z_{r,t_{k}} - c_s(z_{t_{k-1}^+})\Delta\widehat{W}_{\varepsilon}(k),
\end{equation}
thus, in light of the previous discussion,
\begin{equation}
    z_{r,t_k^+} = z_{r,t_{k-1}^+} + \int_{t_{k-1}}^{t_k} \!\!\!\!\!\!\!vd\tau + \int_{t_{k-1}}^{t_k}\! \!\!\!\!\!\! c_s(z_\tau) d\mathcal{W}_\tau - c_s(z_{t_{k-1}^+})\Delta\widehat{W}_{\varepsilon}(k).
\end{equation}
This can be rewritten as
\begin{equation}
   z_{r,t_k^+} = z_{r,t_0} + \int_{t_0}^{t_k} \!\!\!\!\!vd\tau + \int_{t_0}^{t_k}\! \!\!\!\! c_s(z_\tau) d\mathcal{W}_\tau - \sum_{i=1}^k c_s(z_{t_{i-1}^+})\Delta\widehat{W}_{\varepsilon}(i).
\end{equation}
By Lemma~\ref{lemma:brownian_motion_approximation}
\begin{equation}
\lim_{\varepsilon\rightarrow 0} \sum_{i=1}^k c_s(z_{t_{i-1}^+})\Delta\widehat{W}_{\varepsilon}(i) = \int_{t_0}^{t_k}\! \!\!\!\! c_s(z_\tau) d\mathcal{W}_\tau \quad \text{almost surely},
\end{equation}
hence
\begin{equation}
\label{eq:limit_at_sampling_times}
    \lim_{\varepsilon\rightarrow 0}  z_{r,t_k^+} = z_{r,t_0} + \int_{t_0}^{t_k} \!\!\!\!\!vd\tau = z_{r,t_k}^{lin} \quad \text{almost surely.}
\end{equation}
Now, consider any $t \in (t_{k-1},t_k)$, for some $k$, and let $0<\bar \varepsilon<\varepsilon$ be given by $\bar \varepsilon = t-t_{k-1}$. Then
\begin{equation}
    z_{r,t} = z_{r,t_{k-1}^+} + \int_{t_{k-1}}^{t} \!\!\!\!\!\!\!vd\tau + \int_{t_{k-1}}^{t}\! \!\!\!\!\!\! c_s(z_\tau) d\mathcal{W}_\tau,
\end{equation}
which can be discretised with the Euler-Maruyama method as
\begin{equation}
    z_{r,t} = z_{r,t_{k-1}^+} + v(z_{t_{k-1}^+})\bar \varepsilon +  c_s(z_{t_{k-1}^+}) \Delta W_{\bar \varepsilon}(k) + o(\bar{\varepsilon}^2),
\end{equation}
with $\Delta W_{\bar \varepsilon}(k) = \mathcal{W}_t - \mathcal{W}_{t_{k-1}}$, and where $o(\bar{\varepsilon}^2)$ is the one-step truncation error, which is of order $\bar{\varepsilon}^2$. Similarly, the discretised dynamics of $z_{r,t}^{lin}$ is
\begin{equation}
    z_{r,t}^{lin} = z_{r,t_{k-1}}^{lin} + v(z_{t_{k-1}}^{lin})\bar \varepsilon + o(\bar{\varepsilon}^2).
\end{equation}
The difference $z_{r,t} - z_{r,t}^{lin}$ is therefore
\begin{multline}
\label{eq:difference_state_trajectories}
    z_{r,t} - z_{r,t}^{lin} = z_{r,t_{k-1}^+} - z_{r,t_{k-1}}^{lin} + (v(z_{t_{k-1}^+}) - v(z_{t_{k-1}}^{lin}))\bar \varepsilon +\\  c_s(z_{t_{k-1}^+}) \Delta W_{\bar \varepsilon}(k) + o(\bar{\varepsilon}^2)
\end{multline}
Observe that, by the properties of the Brownian motion, $\Delta W_{\bar \varepsilon}(k)$ is a normally distributed random variable with zero expectation and variance $\bar \varepsilon$. Moreover $c_s$ is bounded near zero, because it is the Lie derivative of a smooth function. By taking the limit of~\eqref{eq:difference_state_trajectories} as $\varepsilon$ (hence $\bar \varepsilon$) goes to zero and using~\eqref{eq:limit_at_sampling_times}, we have that for every $\sigma$ 
\begin{equation}
        \lim_{\varepsilon\rightarrow 0} \mathfrak{P}(\left| z_{r,t} - z_{r,t}^{lin} \right| \ge \sigma) = 0,
    \end{equation}
    hence the claim follows.
\end{proof}

The previous proposition states that, when $u^s_t$ is selected as
\begin{equation}
\label{eq:compensations}
    u^s_t =-\sum_{i=1}^{k}\frac{c_s(z_{t_{i-1}^+})\Delta \widehat{W}_\varepsilon(i)}{b(z_{t_i})}\delta(t-t_{i}), \quad t\le t_{k},
\end{equation}
the control $\hat{u}_t^{lin}$ approximates the idealistic control $u_t^{lin}$ as $\varepsilon$ tends to zero, in the sense that the dynamics of the variable $\zeta_t$ can be made approximately linear with an accuracy increasing as $\varepsilon$ decreases.

\begin{remark}
\label{remark:persistence_excitation}
The rationale of Assumption~\ref{assumption:persistence_excitation} is that the noise is persistently exciting, so it can be estimated by measuring the state of the system. From a technical viewpoint, this assumption is necessary for the development of the theory because it ensures the boundedness of $L_{t_k}^\dagger$. However, Assumption~\ref{assumption:persistence_excitation} should not be considered practically restrictive. In fact, if $\vert L_{t_k} \vert < \bar \delta$, for an arbitrarily small $\bar \delta$, this would imply that the noise contribution to the dynamics of the system is sufficiently small at time $t_k$ to be considered negligible. In that case it is simply possible to avoid performing the stochastic compensation at time $t_k$ and still obtain good control performance. 
\end{remark}

\subsection{Remarks on the control input in the diffusion: $m\ne 0$}
\label{sec:noise_in_diffusion}
At the beginning of Section~\ref{section:linearisation} we made the standing assumption that $m(x) \equiv 0$ near $\bar x$, thus restricting the class of systems considered to those with the control input appearing only in the drift term of the stochastic differential equations. This assumption allowed us to obtain sharper analytical results.  In this section we discuss the rationale of such assumption from a control viewpoint, which has implications both in the idealistic and the practical frameworks.

As shown in \eqref{eq:normal_form}, for general $m$, the dynamics of the $r$-th component of the transformed state $z_t = \Phi(x_t)$ is a quadratic function of the control input. This implies that, in the case that $a(z_t) = (1/2)\dL{m}{m}{}\SL{f}{l}{r-1}h(\Phi^{-1}(z_t))$ is not zero for $x$ near $\bar x$, the idealistic feedback linearising control would have the form
\begin{equation}
    u_t^{lin} = \frac{-b(\xi_t,z_t) \pm \sqrt{b(\xi_t,z_t)^2 - 4a(z_t)(c(\xi_t,z_t)-v)}}{2a(z_t)},
\end{equation}
which has real values if and only if the input $v$ is such that $b(\xi_t,z_t)^2 - 4a(z_t)(c(\xi_t,z_t)-v) \ge 0$. The corresponding zero-noise control would have the form 
\begin{equation}
     u_t^{zn,lin} = \frac{-b_d(z_t) \pm \sqrt{b_d(z_t)^2 - 4a(z_t)(c_d(z_t)-v)}}{2a(z_t)},
\end{equation}
where the input $v$ should enforce $b_d(z_t)^2 - 4a(z_t)(c_d(z_t)-v) \ge 0$. It is evident that, although in both cases the control input can be forced to be real by an appropriate choice of $v$, from a practical viewpoint such a choice may not leave any space for a control law achieving objectives such as, for instance, output tracking. Moreover, even though we may be able to define a zero-noise control, the construction of a hybrid controller presents even more challenges, as the quadratic $u_t^{lin}$ is not affine in $\xi_t$. Of course, one might make the standing assumption that $a(z_t) \equiv 0$ for $x$ near $\bar x$. This class of systems would include the systems for which $m\equiv 0$, which we address in detail in this paper, but also those systems with $m\ne 0$ with the additional assumption that $\dL{m}{m}{}\SL{f}{l}{r-1}h(\Phi^{-1}(z_t)) = 0$. However, the design of controllers  poses substantial technical challenges even in this case. For the sake of completeness, we briefly give an overview of these issues. Suppose $m\ne 0$ and $a(z_t) \equiv 0$ for $x$ near $\bar x$. The coefficient $b$ in~\eqref{eq:normal_form} is therefore affine in $\xi_t$, \emph{i.e.}, $b(\xi_t,z_t) = b_d(z_t) + b_s(z_t)\xi_t$. In the idealistic framework, a feedback linearising control would require a division by $b(\xi_t,z_t)$. Unless further, possibly restrictive, assumptions are made, such a division may result in a non-admissible control. Additionally, the fact that the coefficient $b$ depends on the noise affects impulsive compensations as well. In fact, the dynamics of $z_r$ when the control $\hat{u}_t^{lin}$ is applied is
\begin{equation}
    \dot{z}_r = v +c_s(z_t)\xi_t + b_d(z_t)u^s_t + b_s(z_t)u_t^s\xi_t.
\end{equation}
Thus any compensating control $u_t^s$ inevitably introduces noise at time $t_k^+$ which cannot be compensated for because $\xi_{t_k}$ is a random variable which is independent of the process $\xi_t$, $t<t_k$. This makes it impossible to conclude on the convergence in probability of $z_{r,t_k^+}$ to $z_{r,t_k}^{lin}$ as $\varepsilon$ approaches zero, because the random variable $\xi_{t_k}$ can take arbitrarily large values in $\mathbb{R}$ with nonzero probability.

\section{Asymptotic Output Tracking}
\label{sec:stab_track}
In this section we first design an idealistic controller, \emph{i.e.} using the white noise in the feedback loop, to make the output asymptotically track a reference signal. Then we show that a practical feedback control law leveraging the causal stochastic compensations introduced in Section~\ref{sec:practical_linearisation} is able to retrieve the idealistic result, in the limit of the sampling time $\varepsilon$ going to zero.

We start by introducing the idealistic control and we show that, under suitable stability hypotheses on the zero dynamics, it is possible to control the system so that its output tracks reference trajectories while its internal variables remain bounded almost surely. First observe that, as long as $z_t = \Phi(x_t)$ is chosen such that $q\equiv 0$ near zero (which is always possible, see Section~\ref{subsection:sharper_results}), the zero dynamics of system~\eqref{eq:normal_form_restricted} is affine in $\xi_t$. Let it be expressed as
\begin{equation}
\label{eq:zero_dynamics_linear_dependent_noise}
    \dot{\eta}_t = p(\xi_t,0,\eta_t) = p_d(0,\eta_t) + p_s(0,\eta_t)\xi_t.
\end{equation}
Consider a reference signal $y_R$ which is continuously differentiable $r$ times with values in a neighbourhood of zero.  We assume that the initial state of the transformed system~\eqref{eq:normal_form_restricted} is arbitrary while in a neighbourhood of zero and we seek a feedback control $u_t$ that makes the output $y_t$ of the system asymptotically converge to $y_R$. Let
\begin{equation}
    v(\zeta_t, y_R(t)) = y_R^{(r)} - \sum_{i=1}^rd_{i-1}(\zeta_i - y_R^{(i-1)}),
\end{equation}
with $d_i\in \mathbb{R}$ for $i = 0,\dots,r-1$ to be determined, and the admissible idealistic feedback control law be given by 
\begin{equation}
    u_t^{track} = -\frac{c(\xi_t,\zeta_t, \eta_t) - v(\zeta_t,y^R(t))}{b(\zeta_t, \eta_t)}.
\end{equation}
Define the tracking error $e_t := y_t - y_R(t)$. Then the control $u_t^{track}$ forces the dynamics of the tracking error to be 
$e^{(r)}_t + d_{r-1}e^{(r-1)}_t + ... + d_1e^{(1)}_t + d_0e_t$, which can be made exponentially stable by selecting the coefficients $d_i$ such that
\begin{equation}
\label{eq:characteristic polynomyal}
    \Lambda(s) = s^r + d_{r-1}s^{r-1} + ... + d_1s + d_0,
\end{equation}
which is the characteristic polynomial of the matrix
\begin{equation}
\label{eq:controllability_form}
    A = \begin{bmatrix}
    0 & 1 & 0 &\dots & 0\\
    0 & 0 & 1 &\dots & 0\\
    \vdots & \vdots & \vdots & \vdots & \vdots\\
    0 & 0 & 0 &\dots & 1\\
    -d_0 & -d_1 & -d_2 & \dots & -d_{r-1}
    \end{bmatrix},
\end{equation}
has roots with negative real parts. We also study the boundedness of $\zeta_t$ and of the internal variable $\eta_t$ under the control $u_t^{track}$, when $y_R$ and its first $r-1$ derivatives are bounded. Define $\zeta_R(t) = \begin{bmatrix}y_R(t)& \dots& y_R^{(r-1)}(t)\end{bmatrix}^\top$ and $\theta_t = \begin{bmatrix} e_t & \dots & e^{(r-1)}_t \end{bmatrix}^\top$.
Then the following result, the proof of which relies on the definition of strict Lyapunov function and some technical lemmas presented in Appendix B, holds.
\begin{theorem}
\label{theorem:asymptotic_output_tracking}
Consider system~\eqref{eq:normal_form_restricted}. Suppose $y_R(t)$, $y_R^{(1)}(t), \dots, y^{(r-1)}_R(t)$ are bounded. Let $\eta_{R,t}$ be the solution of
\begin{equation}
\label{eq:zero_dynamics_driven}
    \dot{\eta}_{R,t} = p(\xi_t, \zeta_R(t), \eta_{R,t}), \quad \eta_{R,0} = 0
\end{equation}
 and let $p_d$ and $p_s$ be Lipschitz continuous. Moreover, assume that there exists a strict Lyapunov function $V(\eta_{R},t)$ for~\eqref{eq:zero_dynamics_driven} such that $\frac{\partial V}{\partial \eta_{R,i}}(x,t)$ and $\frac{\partial^2 V}{\partial \eta_{R,i} \partial \eta_{R,j}}(x,t)$ are bounded for all $x$ in a neighbourhood of the origin and $t\ge 0$. Suppose that the roots of the polynomial $\Lambda(s)$ in~\eqref{eq:characteristic polynomyal} have negative real part. Then for sufficiently small $\epsilon_R > 0$, if
\begin{equation}
\vert z_i(\bar t) - y_R^{(i-1)}(\bar t)\vert < \epsilon_R, \;\; 1\le i\le r, \qquad \Vert \eta_{\bar t} - \eta_{R, \bar t}\Vert < \epsilon_R,
\end{equation}
 then for all $\epsilon>0$ there exists $\delta>0$ such that
 \begin{multline}
    \vert z_i(\bar t) - y_R^{(i-1)}(\bar t)\vert < \delta \rightarrow \vert z_i(t) - y_R^{(i-1)}(t)\vert < \epsilon, \\ 1\le i\le r, \; \text{for all $t\ge \bar t \ge 0$},
    \end{multline}
    \begin{equation}
    \Vert \eta_{\bar t} - \eta_{R, \bar t}\Vert < \delta \rightarrow \Vert \eta_{t} - \eta_{R, t}\Vert < \epsilon\;\; \text{for all $t\ge \bar t \ge 0$},
\end{equation}
almost surely, \emph{i.e.} the response $z_i$ and $\eta_t$, $t\ge\bar t \ge 0$, of system~\eqref{eq:normal_form_restricted} under the control law $u^{track}_t$ is bounded almost surely.
\end{theorem}
\begin{proof}
System~\eqref{eq:normal_form_restricted} under the control law $u^{track}_t$ can be rewritten in the form
\begin{equation}
\label{eq:tracking_error_system}
        \dot{\theta}_t = A\theta_t, \quad
        \dot{\eta}_t = p(\xi_t, \zeta_R(t) + \theta_t, \eta_t),
\end{equation}
with $A$ given by~\eqref{eq:controllability_form}, which has characteristic polynomial $\Lambda(s)$. Therefore $\theta_t$ has asymptotically stable dynamics. Let $\nu_t = \eta_t - \eta_{R,t}$ and $P(\xi_t, \nu_t, \theta_t, t) = p(\xi_t, \zeta_R(t) + \theta_t, \eta_{R,t} + \nu_t) - p(\xi_t, \zeta_R(t), \eta_{R,t})$. Note that the system
\begin{equation}
\label{eq:time_varying_closed_loop_system}
    \dot{\nu}_t = P(\xi_t, \nu_t, \theta_t, t),\quad
    \dot{\theta}_t = A\theta_t 
\end{equation}
is in the form~\eqref{eq:time_varying_corollary} (see Appendix B). Moreover, system~\eqref{eq:time_varying_closed_loop_system} additionally satisfies the hypotheses of Lemma~\ref{lemma:stability_time_varying} in Appendix B because of the assumption of Lipschitz continuity of $p_d$ and $p_s$ and of the existence of $V$ as in the statement. Then, by Lemma~\ref{lemma:stability_time_varying}, $(0,\eta_{R,t})$ is an almost surely uniformly stable solution of~\eqref{eq:tracking_error_system} and the claimed estimates follow. 
\end{proof}
The previous theorem solves the idealistic local asymptotic output tracking problem, \emph{i.e.} the output $y_t=z_1$ asymptotically converges to $y_R$ whilst the state $z_t$ remains bounded almost surely.

We now focus on practical output tracking. While it can be proved (see Proposition~\ref{proposition:zero_noise_stabilisation} in Appendix C) that under some technical assumptions a zero-noise control is sufficient to asymptotically stabilise the equilibrium at the origin, a controller not performing any sort of compensation for the stochastic disturbances does not guarantee asymptotic tracking. This is because if the system tracks a non-zero reference, then the states will not converge to zero and so the noise will enter the dynamics in a persistent fashion. Hence the noise might drive the states away from the desired trajectory, possibly inducing instability. On the contrary, the control law $\hat{u}_t^{track} = u_t^{zn,track} + u^s_t$, with
\begin{equation}
    u_t^{zn,track} = -\frac{c_d(\zeta_t, \eta_t) - v(\zeta_t,y^R(t))}{b(\zeta_t, \eta_t)}
\end{equation}
and $u_t^s$ given by~\eqref{eq:compensations} is able to prevent this when $\varepsilon$ tends to zero by approximately compensating for the Brownian-induced disturbances, as shown in the following result.

\begin{corollary}
\label{corollary:tracking_with_compensation}
 Consider system~\eqref{eq:normal_form_restricted} and suppose that Assumption~\ref{assumption:persistence_excitation} and the assumptions in Theorem~\ref{theorem:asymptotic_output_tracking} hold. Then the control law $\hat{u}_t^{track}$ is such that $y_t$ converges to $y_R$ in probability and the state $z_t$, $t\ge\bar t \ge 0$, of system~\eqref{eq:normal_form_restricted} is bounded in probability in the limit as $\varepsilon$ approaches zero.
\end{corollary}
\begin{proof}
By Theorem~\ref{theorem:asymptotic_output_tracking} the control $u_t^{track}$ makes the output $y_t$ of~\eqref{eq:normal_form_restricted} asymptotically converge to $y_R$ while keeping the internal states bounded. Let $z_{t}^{track}$ and $z_{t}$ be the the state of system~\eqref{eq:normal_form_restricted} when $u_t = u_t^{track}$ and $u_t = \hat{u}_t^{track}$ are applied, respectively. Then, by Theorem~\ref{theorem:approximate_retrieves_ideal}, for every $\sigma > 0$,
\begin{equation}
        \lim_{\varepsilon\rightarrow 0} \mathfrak{P}(\left| z_t - z_t^{track} \right| \ge \sigma) = 0,
    \end{equation}
and the claim follows.
\end{proof}

\section{Illustrative Example}
\label{sec:example}
In this section we illustrate the validity of the theory by means of a numerical example. Consider the following nonlinear stochastic system in the form~\eqref{eq:system_derivatives_restricted} with
\begin{equation}
    f(x_t) = \begin{bmatrix} s_2(1+x_1)\\
     -2\tan x_2\\
     f_3(x_t)\end{bmatrix}, \; g(x_t) = \begin{bmatrix}
     e^{x_3}\\ 0 \\ e^{x_3}
     \end{bmatrix}, \; l(x_t) = \begin{bmatrix}
        x_1\vspace{1mm}\\ 
    \frac{-2x_1}{c_2}\vspace{1mm}\\
    x_1^2
     \end{bmatrix},
\end{equation}
with $f_3(x_t) = 2x_3 +x_1s_2 - \frac{2s_2x_1^2}{c_2^2}$ and $s_i$ and $c_i$ denoting $\sin{x_i}$ and $\cos{x_i}$ respectively. Let the output of the system be $y_t = x_1 + s_2 - x_3$. We are interested in analysing the system around the origin, \emph{i.e.} we set $\bar x = 0$. The goal is to bring the system to its normal form and to perform asymptotic output tracking. The first step is determining the stochastic relative degree of the system at zero. We set $z_1 = y_t =x_1 + s_2 - x_3$ and we compute its derivative applying It\^o's formula (we omit the procedure for brevity), thus obtaining $dz_1 = (-2x_3 - s_2)dt$. As neither the input $u$ nor the noise appear in this expression, the stochastic relative degree, if defined, is higher than one at the origin. We set $z_2 = -2x_3 - s_2$ and, by computing its derivative, we obtain
\begin{equation}
    dz_2 = \left(2s_2 - 4x_3 - 2x_1s_2 + \frac{6x_1^2 s_2}{c_2^2}  -2e^{x_3}u\right)dt + 4x_1d\mathcal{W}_t.
\end{equation}
The system has therefore stochastic relative degree $r = 2$ at the origin. Moreover,
\begin{align}
    \tilde{c}_d(x_t) &= c_d(\Phi(x_t)) = 2s_2 - 4x_3 - 2x_1s_2 + \frac{6x_1^2 s_2}{c_2^2},\\
    \tilde{c}_s(x_t) &= c_s(\Phi(x_t)) = -2e^{x_3},\\
    \tilde{b}(x_t) &= b(\Phi(x_t)) = 4x_1.\\
\end{align}
Setting $z_3 = x_1 - x_3$ makes the coordinate change $z = \Phi(x)$ a diffeomorphism in a neighbourhood of the origin, with
\begin{equation}
    \Phi(x) = \begin{bmatrix}x_1 + s_2 - x_3 \\ -s_2-2x_3 \\ x_1-x_3\end{bmatrix}.
\end{equation}
In this new set of coordinates the dynamics of the system is given by
\begin{align}
\label{eq:example_normal_form}
    \dot{z}_1 & =z_2,\\
    \dot{z}_2 & = c_d(z_t) + c_s(z_t)\xi_t + b(z_t)u,\\
    \dot{z}_3 &= s_2 - 2x_3 + \frac{2x_1^2s_2}{c_2^2} + 2x_1\xi_t = \tilde{p}(\xi_t, \Phi(x_t)),
\end{align}
which is in the stochastic normal form. The zero dynamics of the system is obtained by equating $z_1 = 0$ and $z_2 = 0$, $z_3 = \eta_t$, which yields $x_1 = (3/2)\eta_t$, $s_2 = -\eta_t$. Replacing these in the third equation in~\eqref{eq:example_normal_form} we get the zero dynamics as follows
\begin{equation}
    \dot{\eta}_t = p(\xi_t, 0, \eta_t) = -2\eta_t + \frac{9\eta_t^3}{2(\eta_t^2 - 1)} + 3\eta_t\xi_t.
\end{equation}
Its first approximation around the origin is $\dot{\eta}_t = -2\eta_t + 3\eta_t\xi_t = A\eta_t + F\eta_t\xi_t$, which is asymptotically stable almost surely because $A - F^2/2 < 0$. Therefore the zero dynamics of the system is almost surely asymptotically stable. We now choose a reference signal of the form $y_R(t) = \beta + \alpha \cos(\omega t)$ and illustrate that $u_t^{track}$ achieves asymptotic output tracking. To this end, we first performed a simulation in the idealistic scenario in which the noise is used in the feedback law. We selected the coefficients $d_0 = 12$ and $d_1 = 7$ in the input $v$ so the characteristic polynomial $\Lambda(s)$ in \eqref{eq:characteristic polynomyal} has roots in $-3$ and $-4$, thus guaranteeing asymptotic stability of the linearised sub-system in the coordinates $\zeta_t$. The reference signal $y_R$ is characterised by $\beta = 0.1$, $\alpha = 0.01$ and $\omega = 5$. The nonlinear stochastic differential equations were integrated using the Euler-Maruyama scheme with period $\Delta t = 10^{-6}$. In Figure~\ref{fig:ideal_control} we show the time history of the state in the coordinates $z_t$ when $u_t^{track}$ is applied. Observe that the first two components (blue and red lines) display, as expected, linear and deterministic behaviours, while the internal variable $z_3$ has noisy dynamics. Because of the properties of the zero dynamics, $z_3$ stays bounded under $u_t^{track}$. Moreover, the component $z_1$ (blue line), which by the definition of normal form is the output $y_t$ of the system, asymptotically converges to the reference signal $y_R$ (purple/dashed line).

\begin{figure}
      \centering
      \includegraphics[width=\linewidth,height=6cm]{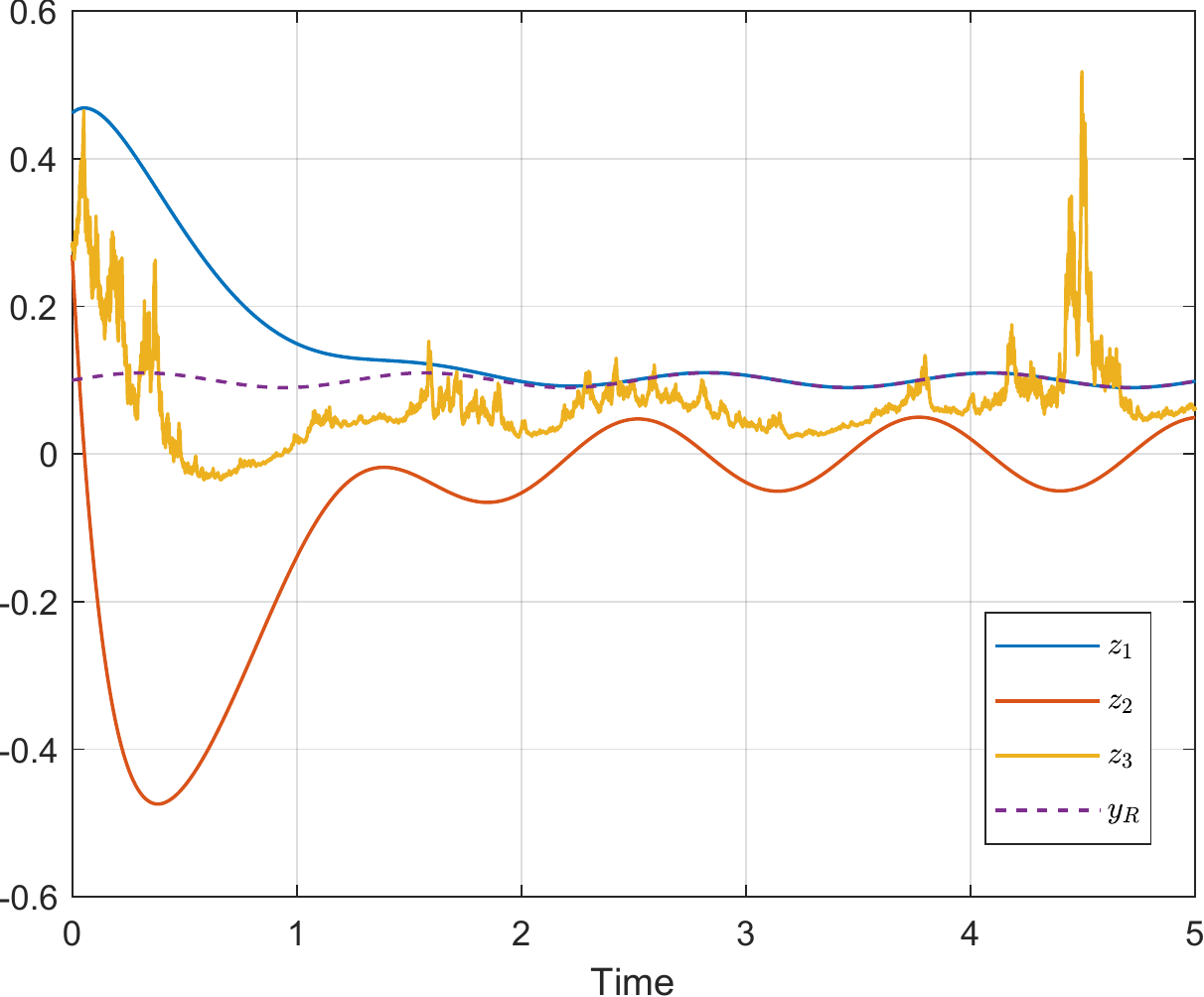}
    \caption{Time history of the state $z_t$ under the control $u_t^{track}$, achieving idealistic output tracking.}
    \label{fig:ideal_control}
\end{figure}

We now turn our attention to practically realisable controllers. To this end, we first discuss the numerical implementation of the simulations in the cases where stochastic compensations are performed with a period $\varepsilon$. The continuous-time dynamics was integrated, as usual, using the Euler-Maruyama numerical scheme with fixed period $\Delta t = 10^{-6}$. The stochastic compensations were performed with a period $\varepsilon$, which must necessarily satisfy $\varepsilon > \Delta t$. We performed simulations with values of $\varepsilon$ of $10^{-3}$, $10^{-4}$ and $10^{-5}$, in order to illustrate the limit behaviour of the solutions as $\varepsilon$ decreases. Note that it could be possible to select smaller $\varepsilon$ as long as $\Delta t$ is decreased accordingly.

We illustrate in Figure~\ref{fig:z_2_tracking} that when the control law with compensations is employed and $\varepsilon$ is decreased, the trajectory of the state under the idealistic control is retrieved. In Figure~\ref{fig:z_2_tracking} we consider again the output tracking setting and display the time history of the component $z_r = z_2$ under the controls $u_t^{zn,track}$ (blue line), $\hat{u}_t^{track}$, with varying $\varepsilon$ ($10^{-3}$ (red line), $10^{-4}$ (yellow line), $10^{-5}$ (purple line)), and $u_t^{track}$ (green line). Observe that, since the control $\hat{u}_t^{track}$ improves the noise compensations as $\varepsilon$ decreases, the trajectories of $z_2$ under the compensated control tend to the idealistic trajectory.

\begin{figure}
    \centering
    \includegraphics[width=\linewidth]{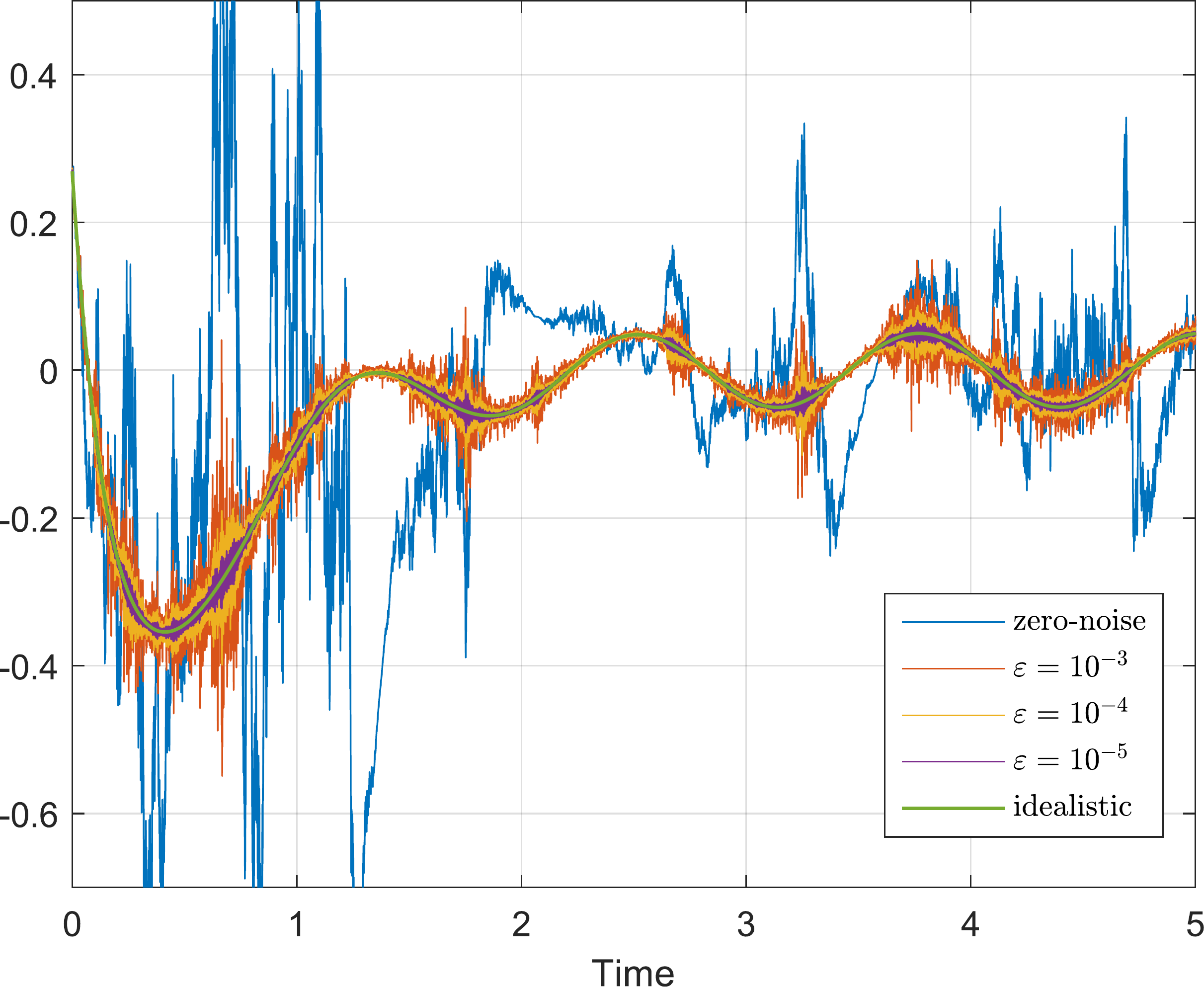}
    \caption{Time history of the component $z_2$ under the controls $u_t^{zn,track}$ (blue line), $\hat{u}_t^{track}$ with $\varepsilon = 10^{-3}$ (red line), $\varepsilon = 10^{-4}$ (yellow line), $\varepsilon = 10^{-5}$ (purple line), and $u_t^{track}$.}
    \label{fig:z_2_tracking}
\end{figure}

 To conclude, Figure~\ref{fig:comparison_tracking} shows a comparison between the time histories of the tracking error $z_1 - y_R = y_t - y_R$ when the controls applied are, respectively, $u_t^{zn,track}$ (blue line) or $\hat{u}_t^{track}$ with decreasing values of $\varepsilon$, namely $10^{-3}$ (red line), $10^{-4}$ (yellow line), $10^{-5}$ (purple line). While the hybrid controller produces, for any of the values of $\varepsilon$, strikingly better asymptotic performances than the zero-noise control, observe that the tracking error is made smaller and smaller as $\varepsilon$ is decreased. This is in line with the fact that the component $z_1 = y_t$ under $\hat{u}_t^{track}$ tends, as $\varepsilon$ approaches zero, to $z_1$ under the control $u_t^{track}$, which in turn asymptotically approaches $y_R$. 
 
 \begin{figure}
    \centering
    \includegraphics[width=\linewidth]{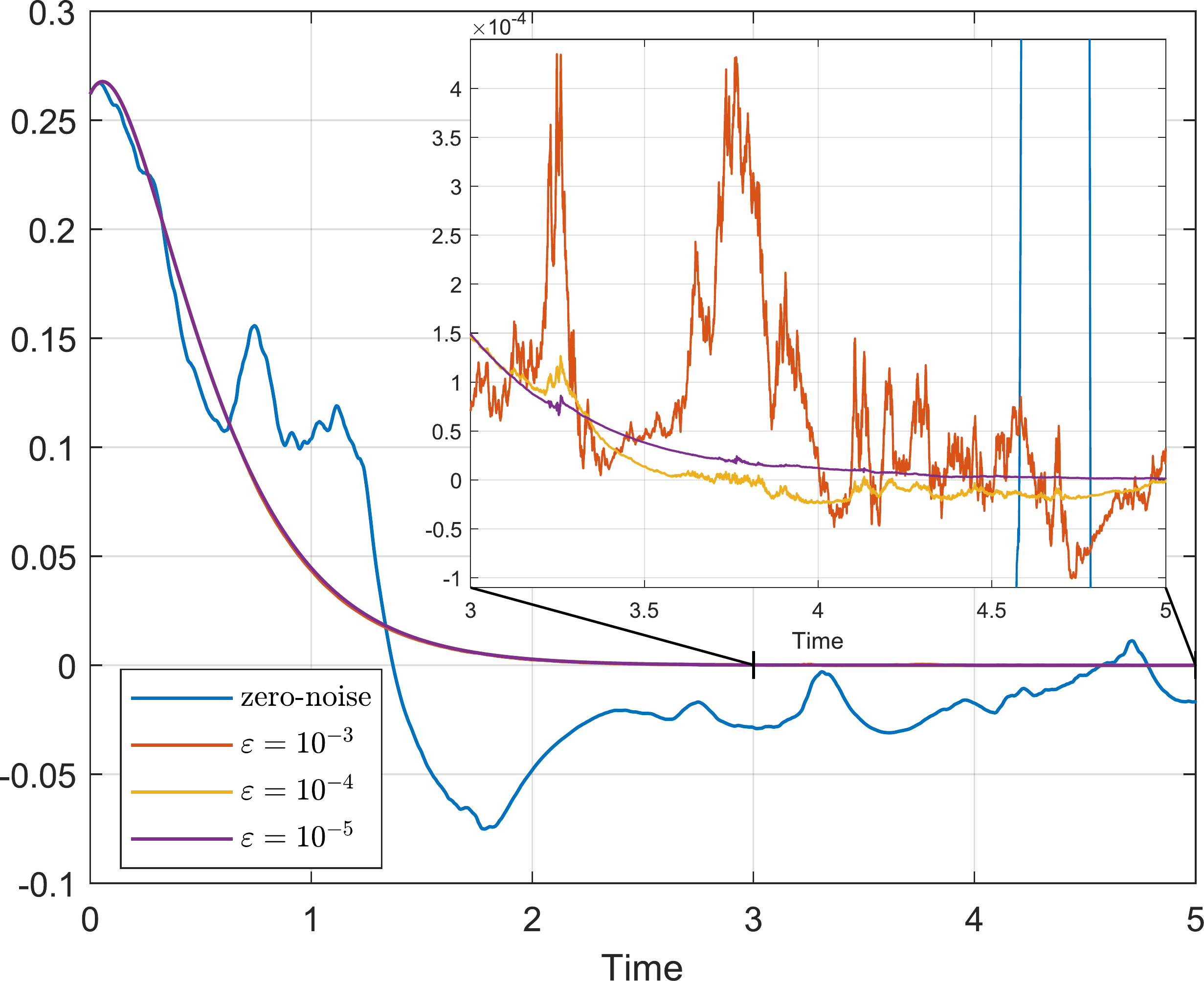}
    \caption{Time history of the tracking error $y_t - y_R$ under the controls $u_t^{zn,track}$ (blue line) and $\hat{u}_t^{track}$ with $\varepsilon = 10^{-3}$ (red line), $\varepsilon = 10^{-4}$ (yellow line) and $\varepsilon = 10^{-5}$ (purple line). }
    \label{fig:comparison_tracking}
\end{figure}

\section{Conclusions}
\label{sec:conclusions}
In this paper we have addressed the path-wise control of nonlinear stochastic systems. First we have introduced a notion of stochastic relative degree and normal form. Then, leveraging these, we have presented a feedback linearising controller. We have observed that this controller is not causal, hence referred to as idealistic, because it requires a feedback of the noise. To overcome this limitation, we have introduced a hybrid control architecture that estimates the Brownian motion from measurements of the state and uses these estimates to periodically compensate, in approximate way, for the noise. We have proved that the performances of the idealistic control law are retrieved when compensations at a frequency tending to infinity are performed. Finally, we have solved the problem of asymptotic output tracking, both in the idealistic and practical framework, and we have provided an illustrative example. 

\section*{Appendix}
\subsection{Technical lemmas and proof of Proposition~\ref{proposition:stochastic_relative_degree}}
\begin{lemma}
\label{lemma:operators_equal_to_zero}
    Let $x\in U\subset \mathbb{R}^n$ and $k\in \{0,...,r-2\}$. Then $\A{l}{m}{g}{}\SL{f}{l}{k}h(\xi_t,x) = 0$ for all $\xi_t \in \mathbb{R}$ if and only if $\Lie_g\SL{f}{l}{k}h(x) +  \dL{l}{m}{}\SL{f}{l}{k}h(x)= 0$ and $\Lie_m \SL{f}{l}{k}h(x) = 0$.
\end{lemma}
\begin{proof}
    Without loss of generality assume $k = 0$. Then the expression of $\A{l}{m}{g}{}\SL{f}{l}{k}h(\xi_t,x)$ for $k=0$ is given in~\eqref{eq:first_order_control_term}. The sufficiency is trivial. As for the necessity, observe that if $\Lie_m h(x) \ne 0$ then the randomness induced by the white noise implies that $\A{l}{m}{g}{}h(\xi_t,x) \ne 0$ almost surely. Therefore $\Lie_m h(x) = 0$ is a necessary condition for $\A{l}{m}{g}{}h(\xi_t,x)$ to be zero for all $\xi_t \in \mathbb{R}$. As a consequence, also $\Lie_g h(x) +  \dL{l}{m}{}h(x)= 0$ is a necessary condition for $\A{l}{m}{g}{}h(\xi_t,x)$ to be zero for all $\xi_t\in\mathbb{R}$.
\end{proof}

\begin{lemma}
\label{lemma:operators_different_than_zero}
    Let $\bar x \in \mathbb{R}^n$. Then $\A{l}{m}{g}{}\SL{f}{l}{r-1}h(\xi_t,\bar x) \ne 0$ almost surely if and only if $\Lie_g\SL{f}{l}{r-1}h(\xi_t,\bar x) +  \dL{l}{m}{}\SL{f}{l}{r-1}h(\xi_t,\bar x)\ne 0$ or $\Lie_m \SL{f}{l}{r-1}h(\xi_t,\bar x) \ne 0$.
\end{lemma}
\begin{proof}
    Observe that $\A{l}{m}{g}{}\SL{f}{l}{r-1}h(\xi_t,\bar x) =  \Lie_g \SL{f}{l}{r-1}h(\bar x)  + \Lie_m \SL{f}{l}{r-1}h(\bar x)\xi_t + \dL{l}{m}{} \SL{f}{l}{r-1}h(\bar x)$. Then the necessity is trivial. As for the sufficiency, observe that if  $\Lie_g \SL{f}{l}{r-1}h(\bar x) +\dL{l}{m}{} \SL{f}{l}{r-1}h(\bar x) \ne 0$ or  $\Lie_m \SL{f}{l}{r-1}h(\bar x) \ne 0$ then the randomness induced by the white noise implies that $\Lie_g \SL{f}{l}{r-1}h(\bar x) +\dL{l}{m}{} \SL{f}{l}{r-1}h(\bar x) \ne  -\Lie_m \SL{f}{l}{r-1}h(\bar x)\xi_t$ almost surely, hence the claim follows.
\end{proof}

\noindent\textbf{Proof of Proposition \ref{proposition:stochastic_relative_degree}.}
We prove the first part of the proposition by induction. Equation~\eqref{eq:kth_derivative} trivially holds for $k=0$. Now, suppose that it holds for any $k\in\{1,\dots,r-2\}$. Then
\begin{multline}
\label{eq:k_1_derivative}
    y_t^{(k+1)} = \SL{f}{l}{k+1}h(\xi_t,x_t) + \A{l}{m}{g}{}\SL{f}{l}{k}h(\xi_t,x_t)u +\\ \frac{1}{2}\dL{m}{m}{}\SL{f}{l}{k}h(\xi_t,x_t)u^2.
\end{multline}
The first term on the right-hand side expands to
\begin{multline}
    \SL{f}{l}{k+1}h(\xi_t,x_t)=\Lie_f \SL{f}{l}{k}h(x_t) + \Lie_l \SL{f}{l}{k}h(x_t)\xi_t +\\ \frac{1}{2}\dL{l}{l}{}\SL{f}{l}{k}h(x_t),
\end{multline}
which, by the first equality in condition (ND), reduces to
\begin{equation}
    \SL{f}{l}{k+1}h(x_t)=\Lie_f \SL{f}{l}{k}h(x_t) + \frac{1}{2}\dL{l}{l}{}\SL{f}{l}{k}h(x_t),
\end{equation}
\emph{i.e.} $\SL{f}{l}{k+1}h(x_t)$ is independent of $\xi_t$ for all $k\in \{0,\dots,r-2\}$. Going back to~\eqref{eq:k_1_derivative}, by Lemma~\ref{lemma:operators_equal_to_zero}, the first two equalities in condition (CD) are equivalent to $\A{l}{m}{g}{}\SL{f}{l}{k}h(\xi_t,x_t) = 0$. Moreover, by the last equality in condition (CD), $\dL{m}{m}{}\SL{f}{l}{k}h(\xi_t,x_t) = 0$. In conclusion, $y_t^{(k+1)} = \SL{f}{l}{k+1}h(x_t)$ for all $k\in \{0,\dots,r-2\}$, which proves that conditions (ND) and (RD) imply~\eqref{eq:kth_derivative}.

To prove~\eqref{eq:r-th_derivative}, observe that by Lemma~\ref{lemma:operators_different_than_zero}, the first two inequalities in condition (RD) are equivalent to $\A{l}{m}{g}{}\SL{f}{l}{r-1}h(\xi_t,\bar x) \ne 0$. Therefore condition (RD) implies that either  $\A{l}{m}{g}{}\SL{f}{l}{r-1}h(\xi_t,\bar x) \ne 0$ or $\dL{m}{m}{}\SL{f}{l}{r-1}h(\bar x)\ne 0$, thus making the control $u(\bar t)$ appear in the expression of the $r$-th derivative of $y_t$. \hfill $\square$

\subsection{Almost sure stability of perturbed stochastic systems}
Consider the stochastic time-varying system
\begin{equation}
\label{eq:time_varying_not_driven}
    \dot{x}_t = f_d(x_t,t) + f_s(x_t,t)\xi_t = \tilde{f}(\xi_t,x_t,t),
\end{equation}
for which we introduce the following Lipschitz assumption \cite{Khalil2002}.
\begin{assumption}
\label{assumption:lipschitz}
$f_d$ and $f_s$ are locally Lipschitz continuous for all $t\ge 0$.

\end{assumption}
We now introduce a concept of almost sure total stability for nonlinear time-varying stochastic system. Consider the following extension of the definition of strict Lyapunov function for stationary stochastic systems given in \cite{Bardi2005} to the case of time-varying stochastic systems.
\begin{definition}
\label{def:strict_Lyapunov_function}
 Consider the autonomous system~\eqref{eq:time_varying_not_driven}. A smooth function $V: U\times \mathbb{R}_{\ge 0} \rightarrow \mathbb{R}$, where $U$ is a bounded domain, is said to be a \emph{strict Lyapunov function} for system~\eqref{eq:time_varying_not_driven} if
\begin{itemize}
    \item $W_1(x) \le V(x,t) \le W_2(x)$ for all $x\in U$, $t\ge 0$, $V(0,\cdot) = 0$ and $W_1$ and $W_2$ are continuous positive definite functions on $U$;
    \item there exists a positive definite function $W_3(x)$ such that $\frac{\partial V}{\partial t} (x,t)+ \Lie_{f_d}V(x,t) + \frac{1}{2}\dL{f_s}{f_s}{}V(x,t) < -W_3(x)$ and $\Lie_{f_s}V(x,t) = 0$ for all $x\in U$, $t\ge 0$.
\end{itemize}
\end{definition}

It is possible to show that the existence of a strict Lyapunov function for system~\eqref{eq:system_derivatives} is sufficient to conclude the almost sure (uniform) asymptotic stability of its equilibrium point (as defined, \emph{e.g.}, in \cite{Kozin1963},\cite{Khasminskii1967}). Namely, the following result is an extension of~\cite[Theorem 2.5]{Bardi2005}. 

\begin{lemma}
\label{lemma:uniform_stability}
Under Assumption~\ref{assumption:lipschitz}, if there exists a strict Lyapunov function for system~\eqref{eq:time_varying_not_driven}, then the equilibrium point at the origin of system~\eqref{eq:time_varying_not_driven} is almost surely uniformly asymptotically stable.
\end{lemma}
The proof follows by replacing in the proof of \cite[Theorem 2.6]{Bardi2005} the time-invariant infinitesimal generator of $V$ with its time-varying version.

\begin{remark}
By dropping the condition that $\Lie_{f_s}V(x,t) = 0$ in Definition~\ref{def:strict_Lyapunov_function}, the results can be reformulated in the context of the weaker stability in probability. See \cite[Theorem 11.2.8]{Arnold1974} for an equivalent of Lemma~\ref{lemma:uniform_stability} in this context.
\end{remark}
Consider now a persistently perturbed version of system~\eqref{eq:time_varying_not_driven}, namely
\begin{equation}
\label{eq:time_varying_driven_differentials}
    dx_t = (f_d(x_t,t) + p_d(x_t,t))dt + (f_s(x_t,t) + p_s(x_t,t))d\mathcal{W}_t
\end{equation}
where $p_d$ and $p_s$ are Lipschitz. In the following we formalise the idea that if $p_d$ and $p_s$ are small enough perturbations, then the stability properties of~\eqref{eq:time_varying_not_driven} are analogous to those of~\eqref{eq:time_varying_driven_differentials}. In other words, we provide an extension to stochastic systems of the concept of total stability presented, \emph{e.g.}, in~\cite[Definition 56.1]{Hahn1967}. To this end, let $\tilde{p}(\xi_t,x_t,t) = p_d(x_t,t) + p_s(x_t,t)\xi_t$.

\begin{definition}
\label{definition:total_stability}
Consider the stochastic systems~\eqref{eq:time_varying_not_driven} and the perturbed system
\begin{equation}
\label{eq:time_varying_driven}
\begin{aligned}
    \dot{x}_t  &= \tilde{f}(\xi_t, x_t, t) + \tilde{p}(\xi_t,x_t,t)\\
               &= f_d(x_t,t) + p_d(x_t,t)\! + (f_s(x_t,t) + p_s(x_t,t))\xi_t.
\end{aligned}
\end{equation}
Suppose that there exists a solution $x_{p,t}(\bar t, \bar x)$ of~\eqref{eq:time_varying_driven}. The equilibrium $x=0$ of~\eqref{eq:time_varying_not_driven} is said to be \emph{totally stable almost surely} if for each $\epsilon > 0$ there exist positive $\delta_1(\epsilon)$, $\delta_2(\epsilon)$ and $\delta_3(\epsilon)$ such that $\Vert x_{p,t}(\bar t, \bar x)\Vert < \epsilon$ for all $t\ge \bar t$ almost surely provided that $\Vert \bar x \Vert < \delta_1$, $\Vert p_d(x,t) \Vert < \delta_2$ and $\Vert p_s(x,t) \Vert < \delta_3$ for all $(x,t)$ satisfying $t \ge \bar t$ and $\Vert x \Vert \le \epsilon$.
\end{definition}

The following lemma provides some properties that a strict Lyapunov function must satisfy in order for system~\eqref{eq:time_varying_not_driven} to have an almost surely totally stable equilibrium at zero.

\begin{lemma}
\label{lemma:total_stability}
Consider systems~\eqref{eq:time_varying_not_driven} and \eqref{eq:time_varying_driven}. Suppose that Assumption~\ref{assumption:lipschitz} holds. If there exists a strict Lyapunov function $V$ for system~\eqref{eq:time_varying_not_driven} such that
\begin{itemize}
    \item $\frac{\partial V}{\partial x_i}(x,t)$ and $\frac{\partial^2 V}{\partial x_i \partial x_j}(x,t)$ are bounded for all $x\in U$ and $t\ge 0$,
    \item $\Lie_{p_s}V(x,t) = 0$ for all $x\in U$ and $t\ge 0$,
\end{itemize}
then system~\eqref{eq:time_varying_not_driven} has an almost surely totally stable equilibrium at $x=0$.
\end{lemma}
\begin{proof}
The proof follows from \cite[Theorem 56.3]{Hahn1967} by noticing that the time derivatives of $V$ for~\eqref{eq:time_varying_not_driven} and \eqref{eq:time_varying_driven} differ by the term
\begin{equation}
  \frac{\partial V}{\partial x}p_d +  (f_s + p_s)^\top\frac{\partial^2V}{\partial x^2}(f_s + p_s), 
\end{equation}
which, by Assumption~\ref{assumption:lipschitz} and by the hypotheses of this lemma, can be made so small that the time derivative of $V$ for~\eqref{eq:time_varying_driven} is negative. The rest of the proof is analogous to that of \cite[Theorem 56.3]{Hahn1967}.
\end{proof}

Finally, the following lemma gives important stability properties of nonlinear time-varying stochastic systems driven by almost surely stable systems.
\begin{lemma}
\label{lemma:stability_time_varying}
Consider the stochastic system
\begin{equation}
\label{eq:time_varying_corollary}
    \begin{aligned}
        \dot{y}_t &= \tilde{a}(\xi_t,y_t,z_t,t) = \tilde{a}_d(y_t,z_t,t) + \tilde{a}_s(y_t,z_t,t)\xi_t,\\
        \dot{z}_t &= \tilde{b}(\xi_t,z_t) = \tilde{b}_d(z_t) +\tilde{b}_s(z_t)\xi_t,
    \end{aligned}
    \end{equation}
    and assume the following.
    \begin{enumerate}
        \item[A1)] $(y,z) = (0,0)$ is an equilibrium of~\eqref{eq:time_varying_corollary} and $\tilde{a}_d$ and $\tilde{a}_s$ are Lipschitz continuous for all $t\ge 0$ in a neighbourhood of zero.
        \item[A2)] There exists a strict Lyapunov function $V(x,t)$ for the first equation of system~\eqref{eq:time_varying_corollary} such that $\frac{\partial V}{\partial x_i}(x,t)$ and $\frac{\partial^2 V}{\partial x_i \partial x_j}(x,t)$ are bounded for all $x\in U$ and $t\ge 0$.
        \item[A3)] The equilibrium $z = 0$ of $\dot{z}_t = \tilde{b}(\xi_t,z_t)$ is stable almost surely.
    \end{enumerate}
    Then the equilibrium $(y,z) = (0,0)$ of~\eqref{eq:time_varying_corollary} is (uniformly) stable almost surely.
\end{lemma}
\begin{proof}
Set $\tilde{f}(\xi_t,y_t,t) = \tilde{a}(\xi_t,y_t,0,t)$ and $\tilde{p}(\xi_t, y_t,t) = \tilde{a}(\xi_t,y_t,z_t(\bar t, \bar{z}),t) - \tilde{a}(\xi_t,y_t,0,t)$, with $z_t(\bar t, \bar{z})$ solution of $\dot{z}_t = \tilde{b}(\xi_t,z_t)$ satisfying $z_{\bar{t}} = \bar z$. Then the first equation in~\eqref{eq:time_varying_corollary} has the form~\eqref{eq:time_varying_driven}. If $\Vert z_{t} \Vert < \epsilon_z$ almost surely for all $t>\bar t$, then, by Assumption A1 there exist $\mu_1>0$ and $\mu_2>0$ such that
\begin{equation}
\begin{gathered}
    \Vert p_d(y, t) \Vert = \Vert a_d(y,z_t(\bar t, \bar{z}),t) - a_d(y,0,t) \Vert < \mu_1 \epsilon_z,\\
    \Vert p_s(y, t) \Vert = \Vert a_s(y,z_t(\bar t, \bar{z}),t) - a_s(y,0,t) \Vert < \mu_2 \epsilon_z,
\end{gathered}
\end{equation}
almost surely for all $y$ in a neighbourhood of zero and all $t\ge \bar t$. Observe that by Assumption A2 the first and second conditions of Lemma~\ref{lemma:total_stability} are satisfied for the first equation in~\eqref{eq:time_varying_corollary}. Therefore, the first equation in~\eqref{eq:time_varying_corollary} is totally stable at $y=0$. By Definition~\ref{definition:total_stability}, for all $\epsilon_y$, there exist $\delta_1>0$, $\delta_2>0$ and $\delta_3>0$ such that $\Vert \bar y \Vert < \delta_1$, $\Vert p_d(y,t)\Vert < \delta_2$ and $\Vert p_s(y,t)\Vert<\delta_3$, for all $(y,t)$ such that $\Vert y \Vert < \epsilon_y$ and $t\ge \bar t \ge 0 $, imply
\begin{equation}
    \Vert y_t(\bar{t}, \bar{y}) \Vert < \epsilon_y \quad \text{almost surely for all $t\ge \bar t$}.
\end{equation}
Finally, by Assumption A3, it is possible to find a $\delta_z$ such that $\Vert \bar z \Vert < \delta_z$ yields $\Vert z_t \Vert < \bar \delta = \min \{\delta_2, \delta_3\}/\max\{\mu_1,\mu_2\}$ for all $t\ge \bar t$, which, for what stated earlier, makes $\Vert p_d(y, t) \Vert < \mu_1\bar \delta \le \delta_2$ and $\Vert p_s(y, t) \Vert < \mu_2\bar \delta \le \delta_3$ be satisfied almost surely. Therefore, for a bounded initial condition $(\bar y, \bar z)$ in a neighbourhood of the equilibrium $(0,0)$, the trajectory $(y_t,z_t)$ is bounded almost surely in a neighbourhood of $(0,0)$, which concludes the proof. 
\end{proof}

\subsection{Almost sure stability under zero-noise control}
\begin{lemma}
\label{lemma:stability_from_internal_dynamics}
Consider the system
\begin{equation}
\label{eq:system_general_internal_dynamics}
\dot{y}_t = \hat{A} y_t + \gamma(\xi_t,y_t,z_t) \quad
\dot{z}_t = \hat{f}(\xi_t,y_t,z_t),
\end{equation}
with $\gamma$ and $\hat{f}$ affine in $\xi_t$ and suppose that $\gamma(\xi_t, 0, z)=0$ for $z$ in a neighbourhood of zero and that
$\frac{\partial \gamma}{\partial y}(\xi_t,0,0)=0.$
If $\dot{z}_t = f(\xi_t,0,z_t)$ has an almost surely asymptotically stable equilibrium at $z=0$ and the eigenvalues of $A$ all have negative real part, then system~\eqref{eq:system_general_internal_dynamics} has an almost surely asymptotically stable equilibrium at $(y,z) = (0,0)$.
\end{lemma}
\begin{proof}
Consider the stochastic subsystem $\dot{z}_t = f(\xi_t,0,z_t)$. Since this has an almost surely asymptotically stable equilibrium at 0 by assumption, by \cite{Boxler1989} it is possible to decompose the system in the Oseledec subspaces $E_c$ and $E_s$ relative to the zero and negative Lyapunov exponents of the linearised dynamics respectively. Let $z_t^c\in E_c$ be the projection of $z_t$ onto $E_c$ and $z_t^s \in E_s$ the projection of $z_t$ onto $E_s$; then for some $F_c$, $F_s$, $g_c$ and $g_s$
\begin{equation}
\label{eq:internal_dynamics_decoupled}
\begin{aligned}
    \dot{z}_t^c &= F_c(\xi_t)z_t^c + g_c(\xi_t,0, z_t^c, z_t^s),\\
    \dot{z}_t^s &= F_s(\xi_t)z_t^s + g_s(\xi_t,0, z_t^c, z_t^s),\\
\end{aligned}
\end{equation}
where the Lyapunov exponents associated to $F_c$ are zero and those associated to $F_s$ are negative. Let $z_t^s = \pi_1(\xi_t, z_t^c)$ be a stochastic center manifold (see \cite{Boxler1989}) at zero for~\eqref{eq:internal_dynamics_decoupled}. Then, by the reduction principle (\cite[Theorem 7.3]{Boxler1989}), since the subsystem $\dot{z}_t = f(\xi_t,0,z_t)$ has an almost surely asymptotically stable equilibrium at $z=0$ by assumption, then the equilibrium at the origin of
\begin{equation}
\label{eq:reduced_dynamics_1}
   \dot{z}_t^c = F_c(\xi_t)z_t^c + g_c(\xi_t,0,z_t^c, \pi_1(\xi_t,z_t^c)) 
\end{equation}
is asymptotically stable almost surely as well. Consider now the full system~\eqref{eq:system_general_internal_dynamics}, for which a center manifold at zero is described by the pair
\begin{equation}
    z_t^s = \chi_1(\xi_t, z_t^c), \quad  y_t = \chi_2(\xi_t, z_t^c).
\end{equation}
By the reduction principle, the dynamics~\eqref{eq:system_general_internal_dynamics} has an almost surely asymptotically stable equilibrium at $(0,0)$ if the reduced system
\begin{equation}
\label{eq:reduced_dynamics_2}
    \dot{z}_t^c = F_c(\xi_t)z_t^c + g_c(\xi_t,\chi_2(\xi_t,z_t^c),z_t^c, \chi_1(\xi_t,z_t^c))
\end{equation}
has an almost surely asymptotically stable equilibrium at 0. It is easy to see that $\chi_2(\xi_t, z_t^c) = 0$, because, by the properties of $\gamma$, the first approximation of the dynamics of $y_t$ reduces to $\dot{y}_t = Ay_t$, which is asymptotically stable by assumption. Then, since $\chi_1(\xi_t, z_t^c) = \pi_1(\xi_t, z_t^c)$, \eqref{eq:reduced_dynamics_2} reduces to \eqref{eq:reduced_dynamics_1}, which has been proved to have an almost surely asymptotically stable equilibrium at zero. Hence system~\eqref{eq:system_general_internal_dynamics} has an almost surely asymptotically stable equilibrium at $(y,z) = (0,0)$.
\end{proof}

Consider system~\eqref{eq:normal_form_restricted} and the control
\begin{equation}
    u_t^{zn,stab} = -\frac{c_d(\zeta_t, \eta_t) - v(\zeta_t)}{b(\zeta_t, \eta_t)},
\end{equation}
with $c_d$ (and $c_s$) as defined in Section~\ref{subsection:compensating_control} and $v(\zeta_t) = - d_{r-1}\zeta_r - d_{r-2}\zeta_{r-1} - \dots - d_0 \zeta_1$,
with $d_i \in \mathbb{R}$ for $i=1,\dots,r-1$ such that the polynomial~\eqref{eq:characteristic polynomyal} has roots with negative real part. Then the following holds.

\begin{proposition}
\label{proposition:zero_noise_stabilisation}
Consider system~\eqref{eq:normal_form_restricted} and suppose that the equilibrium at $\eta = 0$ of the zero dynamics of system~\eqref{eq:normal_form_restricted} is locally asymptotically stable almost surely. Additionally, suppose that $c_s(0, \eta) \equiv 0$ for $\eta$ in a neighbourhood of zero and that $\frac{\partial c_s}{\partial \zeta}(0, 0) = 0$. Then the control law $u_t^{zn,stab}$ renders the equilibrium $(\zeta,\eta) = (0,0)$ asymptotically stable almost surely.
\end{proposition}
\begin{proof}
Observe that, if $u_t= u_t^{zn,stab}$ in~\eqref{eq:normal_form_restricted} then, the closed-loop dynamics is given by
    \begin{equation}
\label{eq:time_varying_closed_loop_system_approximate}
    \dot{\zeta}_t = A\zeta_t + Bc_s(\zeta_t, \eta_t)\xi_t, \quad \dot{\eta}_t = p(\xi_t, \zeta_t, \eta_t)
\end{equation}
where $A$ is given by~\eqref{eq:controllability_form}, which has characteristic polynomial \eqref{eq:characteristic polynomyal}. This system is in the form~\eqref{eq:system_general_internal_dynamics} with $\gamma(\xi_t,\zeta_t,\eta_t) = Bc_s(\zeta_t, \eta_t)\xi_t$ and, by the assumptions on $c_s$, $\gamma$ satisfies the hypotheses of Lemma~\ref{lemma:stability_from_internal_dynamics}. Therefore, the equilibrium at the origin is asymptotically stable almost surely. 
\end{proof}

\begin{remark}
When studying asymptotic stabilisation, the convergence of the states to the equilibrium at the origin contradicts the persistence of excitation condition of Assumption~\ref{assumption:persistence_excitation}. However, following from the discussion in Remark~\ref{remark:persistence_excitation}, if the hypotheses of Proposition~\ref{proposition:zero_noise_stabilisation} are satisfied, one might perform practical asymptotic stabilisation using a control $\hat{u}_t^{stab} = u_t^{zn,stab} + u_t^s$ until the states are in an arbitrarily small neighbourhood of zero, and then switch to just the zero-noise control $u_t^{zn,stab}$. Doing this has the advantage of obtaining state trajectories that are closer, as $\varepsilon$ goes to zero, to the idealistic ones when the system is away from the equilibrium, thus ensuring a more predictable behaviour of the system.    
\end{remark}

\bibliographystyle{IEEEtran}
\bibliography{bibliography}

\begin{IEEEbiography}[{\includegraphics[width=1in,height=1.25in,clip,keepaspectratio]{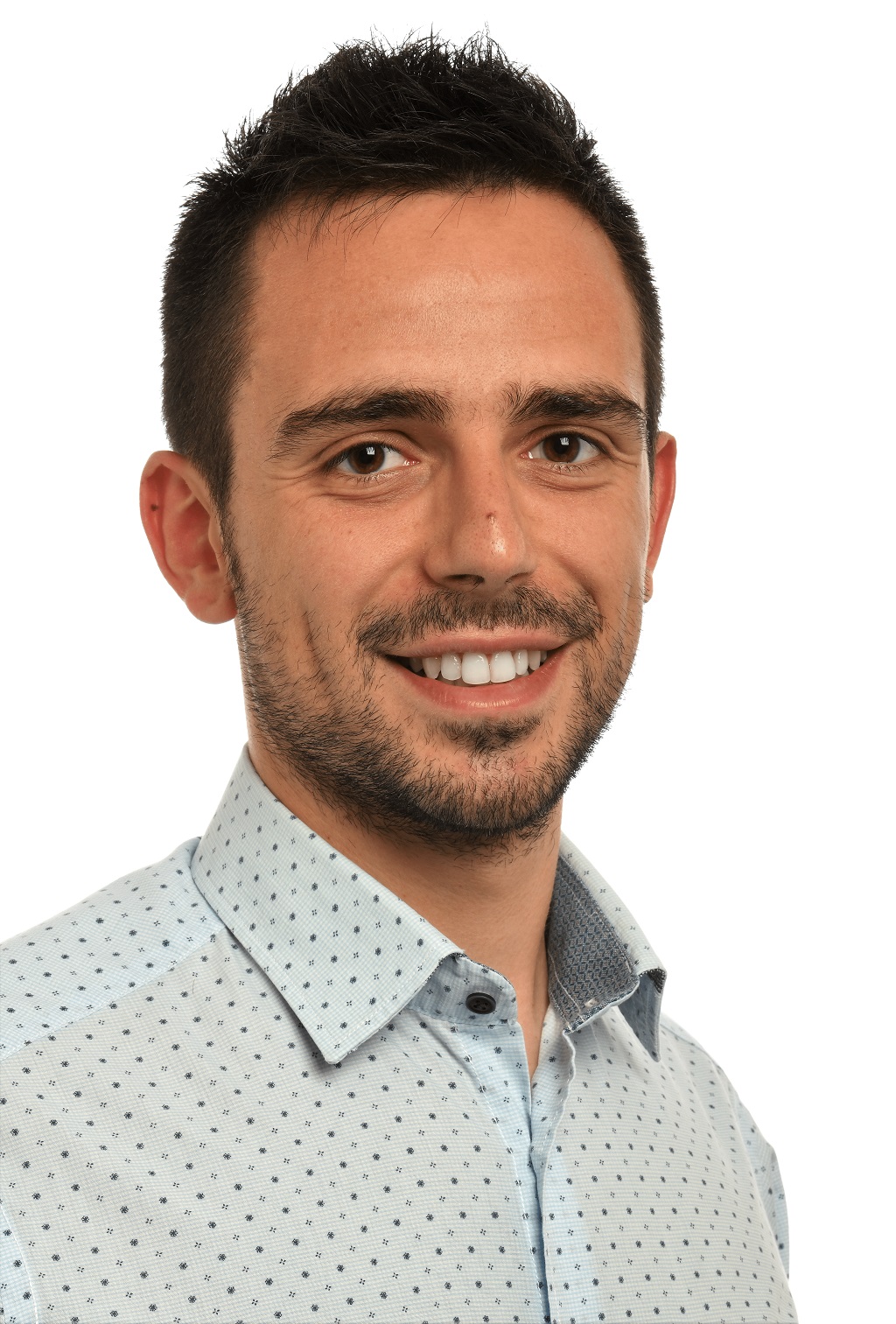}}]{Alberto Mellone} was born in Galatina (Lecce), Italy, in 
1993. He received the B.S. degree in information engineering from the University of Salento, Lecce, Italy in 2015 and the M.S. degree in robotics and automation engineering from 
the University of Pisa, Pisa, Italy in 2018. In 2018 he joined the Control and Power Group, Imperial College London, London, UK, where he he pursued a Ph.D. looking into analysis and control of stochastic systems. He graduated in 2022 with a thesis on path-wise output regulation and nonlinear control of stochastic systems. 
\end{IEEEbiography}

\begin{IEEEbiography}[{\includegraphics[width=1in,height=1.25in,clip,keepaspectratio]{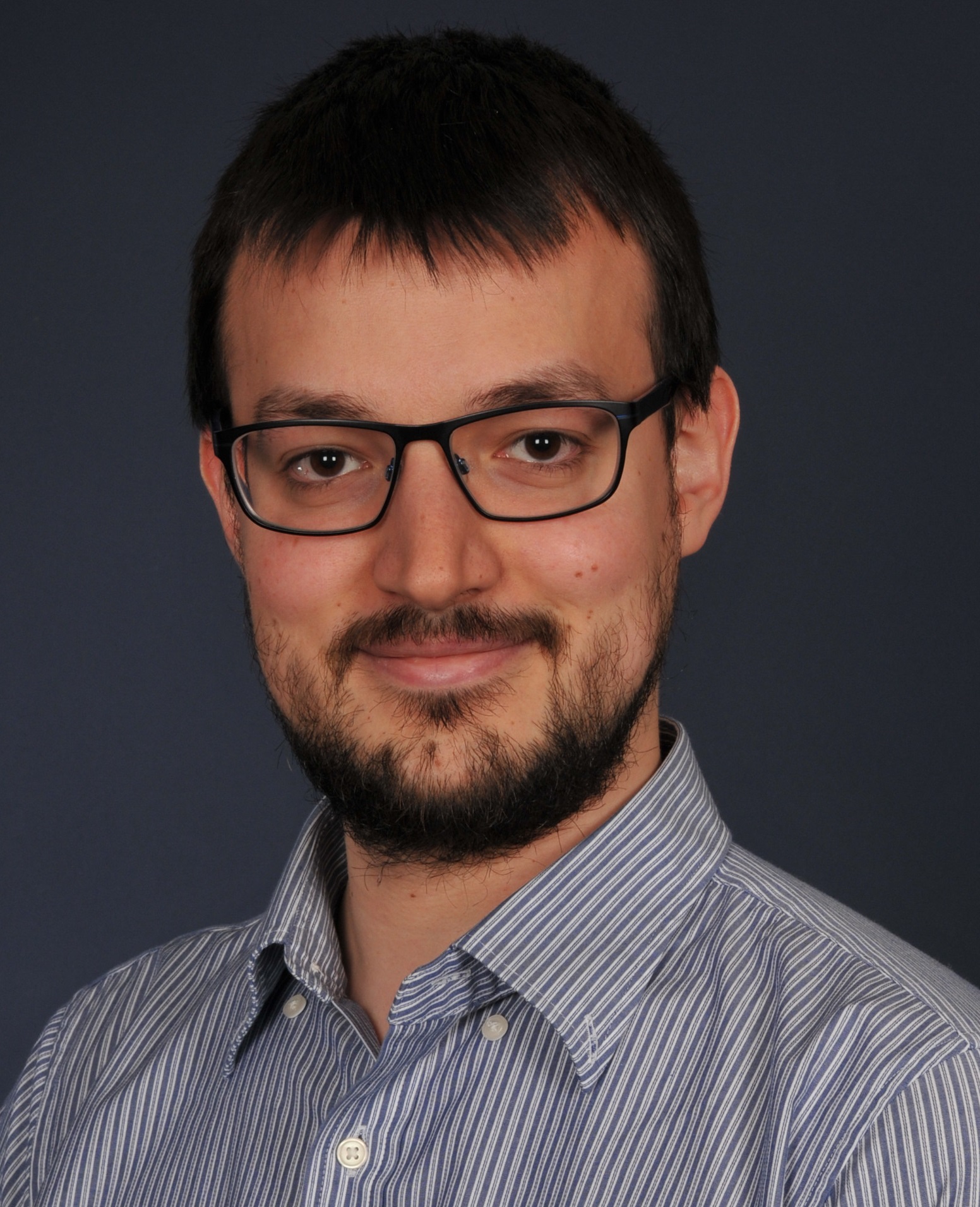}}]{Giordano Scarciotti} (Senior Member, IEEE) was born in Frascati (Rome), Italy, in 1988. He received his B.Sc. and M.Sc. degrees in Automation Engineering from the University of Rome “Tor Vergata”, Italy, in 2010 and 2012, respectively. In 2012 he joined the Control and Power Group, Imperial College London, UK, where he obtained a Ph.D. degree in 2016 with a thesis on approximation, analysis and control of large-scale systems. He also received an M.Sc. in Applied Mathematics from Imperial College in 2020. He is currently a Senior Lecturer in the Control and Power Group. His current research interests are focused on analysis and control of uncertain systems, model reduction and optimal control. He was a visiting scholar at New York University in 2015 and at University of California Santa Barbara in 2016 and he is currently a Visiting Fellow of Shanghai University. He is the recipient of an Imperial College Junior Research Fellowship (2016), of the IET Control \& Automation PhD Award (2016), the Eryl Cadwaladr Davies Prize (2017) and an ItalyMadeMe award (2017). He is a member of the IEEE CSS Conference Editorial Board, of the IFAC and IEEE CSS Technical Committees on Nonlinear Control Systems and has served in the International Programme Committee of multiple conferences. He is the National Organising Committee Chair for the EUCA European Control Conference 2022 and the Invited Session Chair for IFAC Symposium on Nonlinear Control Systems 2022.
\end{IEEEbiography}
\end{document}